\documentclass[11pt]{article}
\usepackage[utf8]{inputenc}
\usepackage{times}
\usepackage[backend=bibtex,giveninits=true, maxbibnames=5]{biblatex}
\usepackage{biblatex}
\usepackage{biblatex}
\usepackage{subfig}
\usepackage[normalem]{ulem}
\usepackage{float}
\usepackage{cancel}
\usepackage{algorithm, changepage}
\usepackage[noend]{algpseudocode}

\bibliography{cc-metric}

\usepackage{setspace}

\usepackage{soul}

\usepackage{nicefrac}

\usepackage[margin=1in]{geometry}
\usepackage{comment}

\usepackage{mathtools}

\usepackage{amsmath,amsthm,amssymb,latexsym,amsfonts, dsfont}
\usepackage{thmtools}
\usepackage{graphicx}
\usepackage[export]{adjustbox}
\usepackage[font=footnotesize,labelfont=bf]{caption}
\usepackage{hyperref}
\usepackage{lineno}
\usepackage{amsthm}
\usepackage{bm}
\usepackage{listings}
\usepackage[nottoc,numbib]{tocbibind}
\usepackage{soul}
\usepackage{thm-restate}
\usepackage{xcolor}

\definecolor{mypink}{RGB}{220,38,127}
\definecolor{myblue}{RGB}{100,143,255}
\definecolor{mygreen}{RGB}{74,167,103}
\definecolor{myorange}{RGB}{254,97,0}
\definecolor{myyellow}{RGB}{255,176,0}
\definecolor{mypurple}{RGB}{120,94,240}

\newtheorem{theorem}{Theorem}
\newtheorem{case}{Case}
\newtheorem{claim}{Claim}
\newtheorem{subcase}{Subcase}[case]

\newtheorem{definition}{Definition}

\newtheorem{lemma}{Lemma}

\newtheorem{fact}{Fact}

\allowdisplaybreaks

\DeclareCiteCommand{\citeyearpar}
    {}
    {\mkbibparens{\bibhyperref{\printdate}}}
    {\multicitedelim}
    {}

\title{Online Correlation Clustering with Metric Weights}

\author{Sami Davies\thanks{Department of EECS, UC Berkeley.} \and Benjamin Moseley\thanks{Carnegie Mellon University.} \and Heather Newman\thanks{Department of Computer Science, Vassar College.}}
\begin{document}
\date{}

\maketitle

\begin{abstract}
The standard online version of correlation clustering is prohibitively hard, as even randomized algorithms cannot achieve competitive ratio better than $\Omega(n)$. Prior works bypass this lower bound by relaxing the online model through recourse, random arrival order, or seeding the algorithm with an offline sample of the underlying input. We instead ask whether additional structure in the input itself can overcome this lower bound.
We study weighted correlation clustering under probability constraints, where $w^+_{uv}+w^-_{uv}=1$ for every $uv$ edge, and triangle inequality constraints, where the negative weights $w^-$ satisfy triangle inequality. While this version of correlation clustering is well-studied in the offline setting, we initiate its online study and give a deterministic online algorithm that maintains a clustering whose total weighted disagreement cost is within an $O(1)$ factor of the offline optimum, against adversarial arrival order. This is the first constant-competitive online algorithm for a natural minimization variant of correlation clustering in the fully online model, and shows that metric consistency on the edge weights separates tractable from intractable online instances.
\end{abstract}

\section{Introduction}

Correlation clustering, introduced by Bansal, Blum, and Chawla~\cite{BBC04}, is a
fundamental problem in unsupervised learning. The input is a complete graph whose edges
are labeled positive (similar) or negative (dissimilar), and the goal is
to partition the vertices into clusters that agree as much as possible with respect to these labels.
The objective is to find a clustering that minimizes the number of edges in disagreement, where an edge is in \emph{disagreement} with respect to a clustering if it is a positive edge
whose endpoints are in different clusters or if it is a negative edge whose
endpoints are in the same cluster. The problem, which we will call \emph{unweighted correlation clustering}, is NP-hard~\cite{charikar2005clustering}, though
it admits constant-factor approximation algorithms and has found broad applications in
information retrieval, entity resolution, community detection, and automated
labeling~\cite{charikar2005clustering,ACN-pivot,GionisMT07}.

In many real-world scenarios, data arrives sequentially, and an element must be assigned
to a cluster immediately after it arrives, without any knowledge of future data. This motivates the
online model for correlation clustering, which has been studied for over fifteen
years~\cite{MathieuSS10,Cohen-AddadLMP22,DMN-online}. In the online model,
nodes arrive one at a time, and when a node $u$ arrives, it reveals its labels on edges $uv$, for nodes $v$ that have already arrived. The algorithm must irrevocably assign each arriving node
to a cluster (either an existing cluster, or a new one containing that node).

This setting is subject to a strong lower bound. Consider two
vertices $u$ and $v$ with a positive edge between them. The optimal offline
solution clusters them together at zero cost. The online algorithm, however, cannot know
whether these two will eventually be part of a single clique of positive edges, or part of two positive cliques with almost no similarities between them.  If the algorithm does not put $u$ and $v$ together, then a positive clique is revealed and the algorithm has infinite competitive ratio. Alternatively, if $u$ and $v$ are put together, then an adversary can reveal an instance in which $u$ and $v$ belong to separate positive cliques, where all edges between the cliques are negative except for $uv$. In this instance, the optimal cost is 1, but the algorithm has cost $\Omega(n)$.  Mathieu, Sankur, and Schudy~\cite{MathieuSS10}
formalized this and proved that no online algorithm (randomized or deterministic) can achieve a competitive ratio better than $\Omega(n)$ for minimizing
disagreements. This is matched, up to constants, by a greedy algorithm. 
Overall, the standard online setting for correlation clustering is essentially hopeless.

In response to this strong lower bound, further work has relaxed the standard online model. One direction allows the
algorithm recourse: a node may be reassigned to a different cluster after its
initial placement, at some cost. Cohen-Addad, Lattanzi, Maggiori, and
Parotsidis~\cite{Cohen-AddadLMP22} showed that $O(\log n)$ recourse per node
suffices to maintain a constant-factor approximation at all times, and proved this recourse bound is tight for maintaining a constant-factor approximation.
Another direction is to change the arrival model. For instance, for random order arrivals, the classic Pivot algorithm of Ailon, Charikar, and
Newman~\cite{ACN-pivot} is $3$-competitive. The authors in~\cite{LattanziMVWZ21} show that the Pivot algorithm is also successful in another beyond-worst-case model known as the \emph{online-with-a-sample} (AOS) model, where the online algorithm is provided with a random $\varepsilon$-fraction of the underlying input as an offline sample upfront and the remainning nodes arrive online adversarially; they achieve an $O(\nicefrac{1}{\varepsilon})$-competitive algorithm in this setting.

These results represent meaningful progress, but each relaxation either surrenders consistency guarantees (recourse) or
makes assumptions on the arrival distribution (random order, or given an offline sample). A different question is
whether a structural restriction on the underlying input itself can circumvent the lower bound without \emph{any} relaxation of the online
model.

\medskip
\noindent \textbf{Weighted Probability Constraints:} Alongside the unweighted setting described above, where edges are labeled positive or negative, there is a steady stream of work on
weighted correlation clustering under \emph{probability constraints} in the offline setting. Here,
each edge $uv$ has a similarity weight $ 0 \leq w^+_{uv} \leq 1$ and
a dissimilarity weight $0 \leq w^-_{uv} \leq 1$, where $w^+_{uv} + w^-_{uv} = 1$. The
 objective is then to minimize the total weighted disagreements, i.e., find a clustering $\mathcal{C}$ minimizing
\[
    \mathrm{cost}(\mathcal{C})
    \;=\;
    \sum_{\substack{uv:\, u,v \text{ same cluster}}} w^-_{uv}
    \;+\;
    \sum_{\substack{uv:\, u,v \text{ diff.\ clusters}}} w^+_{uv}.
\]
Note that if $w^+_{uv}, w^-_{uv} \in \{0,1\}$ for every pair $u,v$, then we recover the
 unweighted problem, while the general case captures fractional or probabilistic similarity
information. A well-motivated further restriction is that the dissimilarity
weights satisfy the triangle inequality, with
$
    w^-_{uz} \;\leq\; w^-_{uv} + w^-_{vz}$.
This encodes a natural transitivity of dissimilarity, where if $u$ is dissimilar to $v$
and $v$ is dissimilar to $z$, then $u$ should be reasonably dissimilar to $z$. For brevity, we will call this problem of weighted correlation clustering, where the weights satisfy probability constraints and the negative weights satisfy triangle inequality, \emph{metric correlation clustering}. The problem is well-motivated. For instance, the clustering aggregation
problem, where one seeks a consensus clustering over a collection of
input clusterings, is NP-hard \cite{BarthelemyL93} and reduces to metric correlation clustering~\cite{GionisMT07}.

Ailon, Charikar, and Newman~\cite{ACN-pivot} gave both a $2.5$-approximation via LP
rounding and a combinatorial $5$-approximation for weighted correlation clustering under probability constraints; both
improve to a $2$-approximation for metric correlation clustering (itself an improvement of a deterministic 3-approximation for metric correlation clustering~\cite{GionisMT07} due to Gionis, Mannila, and Tsaparas, although the former is randomized). Chawla, Makarychev, Schramm, and
Yaroslavtsev~\cite{chawla2015near}
obtained a $2.06$-approximation for the unweighted problem, and a $1.5$-approximation (via LP rounding) for metric correlation clustering. In that work, the they also prove an approximation preserving reduction from the weighted problem under probability constraints to unweighted correlation clustering; thus the best approximation factor for metric correlation clustering is now below 1.5 \cite{cao2024understanding}.
 Recently,
Ostovari and Zarei~\cite{OstovariZ25} gave a combinatorial $O(n^2)$-time
$1.6$-approximation for metric correlation clustering.

Despite this offline attention, metric correlation clustering has received
no study in the online setting. A natural conjecture may be that the $\Omega(n)$
lower bound on the competitive ratio for unweighted correlation clustering due to Mathieu, Sankur, and Schudy~\cite{MathieuSS10} extends to this model, making it just
as intractable online. However, we observe this is false.
The lower bound construction of~\cite{MathieuSS10} for online unweighted correlation clustering relies critically on \emph{bad triangles}: three vertices between which two edges are positive and one is negative.
However, instances with bad triangles violate the triangle inequality \emph{enforced on the negative weights}, so the adversary's ability to confuse an online algorithm relies on
presenting dissimilarity patterns that are metrically inconsistent and thus does not extend to the problem we consider here. Still, it is totally unclear what competitive ratio is possible with probability constraints subject to metric weights. We thus ask the question:
\begin{center}
     \emph{Is online correlation clustering fundamentally easier under metric weights? In particular, does online metric correlation clustering admit a constant-competitive algorithm?}
\end{center}

\subsection{Our results}
We answer this question affirmatively. We present the first online algorithm for
metric correlation clustering
that is constant-competitive. Our algorithm is \emph{purely online}: the input arrives in adversarial order and we do not require recourse or other beyond-worst-case approaches.

\begin{theorem}
\label{thm:main}
There exists a deterministic online algorithm for metric correlation clustering that outputs a clustering whose cost is at most $O(1)$ times the cost of the
optimal offline clustering.
\end{theorem}

This is the first $O(1)$-competitive online algorithm for any natural minimization variant of correlation clustering without relaxing the online model. Our result demonstrates that
enforcing the triangle inequality on the negative weights separates tractable from intractable
instances of online correlation clustering.

\subsection{Related work}

Unweighted correlation clustering was introduced by
Bansal, Blum, and Chawla \cite{BBC04}.  As the problem is NP-hard \cite{charikar2005clustering}, approximation algorithms have been steadily improving over the past two decades \cite{ACN-pivot, chawla2015near, cohen2022correlation}, with the current best approximation being 1.485 \cite{cao2024understanding}. 

Recall there is no constant-competitive algorithm in the purely online setting for unweighted correlation clustering \cite{MathieuSS10}, though relaxed online models including allowing recourse \cite{cohen2022correlation} and seeding the online algorithm with an offline sample (the AOS model) \cite{LattanziMVWZ21} have been studied. We note the AOS model has also been used for the more general $\ell_p$-norm objectives in correlation clustering \cite{DMN-online}, and recourse has been used to study correlation clustering with additional fairness constraints \cite{balkanski2025faironlineCC}.

Metric correlation clustering is well-studied offline \cite{GionisMT07, ACN-pivot, chawla2015near, OstovariZ25}, with the current best approximation factor being the same as the best factor for the unweighted problem due to a reduction by Chawla et al. \cite{chawla2015near}. The best factor from a combinatorial algorithm is $1.6.$ \cite{OstovariZ25}. Previous techniques are seemingly not amenable to a purely online setting; while LP based algorithms are clearly not amenable here, even the combinatorial algorithms, which are all pivot-like algorithms, turn out to fail under adversarial arrival order (we discuss this more in Section \ref{sec:tech-overview}). Moreover, our main insight for online metric correlation clustering is that \emph{cluster centers should move} in order to obtain constant-competitiveness, and (to the best of our knowledge) no previous algorithms for correlation clustering allow this.

\section{Technical Overview}\label{sec:tech-overview}

Throughout, we write $d_{uv} = w^-_{uv}$ for the dissimilarity weight of a pair.
By the triangle inequality assumption on $w^-$, the function $d$ is a semi-metric
on $V$ (with $d_{uu}=0$), and the similarity weight of a pair is
$w^+_{uv} = 1 - d_{uv}$. We design the algorithm around distances with respect to $d$: small
$d_{uv}$ means $u$ and $v$ are similar and ``want'' to be clustered together,
while large $d_{uv}$ means they want to be separated. We write $\mathcal{C}$
for our algorithm's clustering and $\mathcal{C}^*$ for a fixed optimal one.

The constant-factor approximations for metric correlation clustering are all
\emph{pivot-based}, meaning they repeatedly select some unclustered vertex $p$ as the pivot / center, then define rules for adding unclustered $u$ to $p$'s cluster (e.g., add $u$ to $p$'s cluster if $d_{up}$ is less than some threshold, or with probability  $f(d_{up})$ for $f$ some carefully chosen function). Importantly, in all of these works, a single node (the pivot / center) is the sole node used to determine which other nodes join its cluster.
We build intuition for our algorithm by describing how
these pivot-based strategies break in the purely online, adversarial-order model, and the
combinatorial ideas we must develop in response. 

First, we explain that the most natural online Pivot algorithm (the threshold variant) fails because a center can become a poor representative
of its own cluster. Second, it is tempting to adapt the density-based algorithm of Charikar, Guruswami,
and Wirth (CGW) \cite{charikar2005clustering} because it chooses pivots in arbitrary order. However, its natural
online adaptations also fail, because one cannot know in advance which future
nodes will be close to, or far from, a pivot. Then, we overview our algorithm and our charging procedures.

\paragraph{Why online Pivot fails.}
The natural online Pivot algorithm fixes a radius $r \in (0,1)$ and, when $u$
arrives, places $u$ in the cluster of the earliest-arriving active pivot $p$
with $d_{up}\le r$; if no such $p$ exists, $u$ starts a new cluster and is the pivot of that cluster.\footnote{Here $r$ is
an arbitrary radius used only to illustrate the failure of online Pivot. This should not be confused with the parameter  $r$ that our algorithm specifies later.} We show this fails on an instance of metric correlation clustering on the line segment $[0,1]$ (see the left-hand panel of Figure \ref{lineexample}). Node $p$
arrives first at point 0. Then, $k$ nodes $S_1$ arrive at point $r$, so online Pivot puts all of $S_1$ into $p$'s cluster.
Next, $k$ nodes $S_2$ arrive at point $r+\varepsilon$, so they just miss being clustered by
$p$. Online Pivot opens a second cluster at the first node to arrive of
$S_2$, and all of $S_2$ is clustered together. With $\varepsilon = \nicefrac{1}{k}$, the algorithm pays $\Omega(k^2)$ for the positive
edges between $S_1$ and $S_2$, whereas a single cluster on the whole instance
costs (for fixed constant $r$) only $O(k)$ in negative weight. Choosing
$k = \Theta(n)$ gives a $\Omega(n)$ lower bound on this class of algorithms. It is also not hard to construct slightly different instances to see that this difficulty cannot be overcome by choosing $r$ randomly, or by assigning each pivot its own threshold.

Once many nodes of $S_1$ have arrived, $p$ is no longer a
good representative of the cluster. 
An algorithm has two choices: either stop clustering nodes into $p$'s cluster, or choose a new pivot for the cluster. In this case,
a suitable response is to update
the pivot from $p$ to a node of $S_1$ and make subsequent clustering decisions
from there. We use inspiration from algorithms with density-based pivots to help us decide whether a cluster's pivot should move, versus when it should stop accepting new nodes.

\paragraph{Density-based pivots.}
For the online setting, we need an algorithm whose guarantee does not depend on
the order in which centers are chosen. The density-based offline algorithm of
Charikar, Guruswami, and Wirth \cite{charikar2005clustering} is a natural choice. Offline, it repeatedly
picks an arbitrary unclustered node $u$, and considers the ball
$\text{Ball}(u,\nicefrac12)$ of unclustered nodes within radius $\nicefrac12$ around $u$. If
the average distance from $u$ to nodes in the ball is below $\nicefrac14$, the algorithm takes $\{u\}\cup \text{Ball}(u,\nicefrac12)$ as a cluster, and otherwise it makes $u$ a
singleton cluster. Then it removes the cluster from the graph and recurses.

\begin{figure}[t]
\centering
\begin{minipage}[b]{0.48\linewidth}
  \centering
  \includegraphics[width=\linewidth]{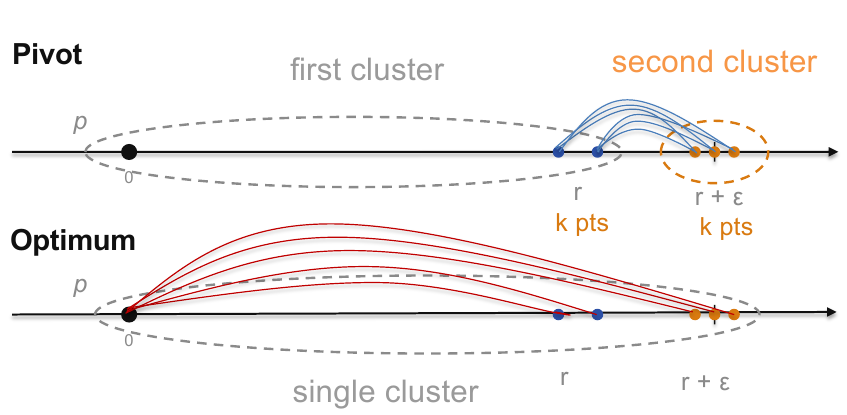}
\end{minipage}%
\hfill
\begin{minipage}[b]{0.48\linewidth}
  \centering
  \includegraphics[width=\linewidth]{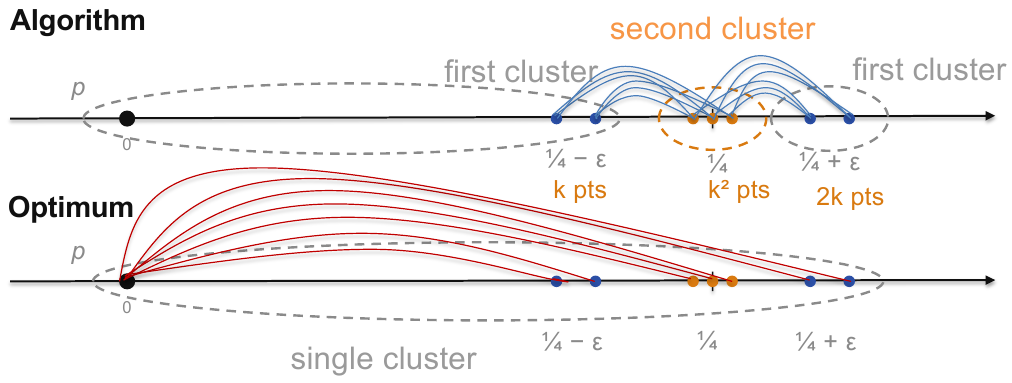}
\end{minipage}
\caption{Lower bounds for two natural online algorithms on metric correlation clustering instances.
\textbf{(Left)} Online Pivot assigns $k$ points at distance $r$ from $p$ to $p$'s cluster, then opens a new cluster for $k$ arriving points at distance $r+\varepsilon$. Pivot separates the two groups, paying $\Omega(k^2)$ for positive edges between them (solid blue edges); the optimum (bottom) clusters all points together, paying only $O(k)$ in negative weight to $p$ (solid red edges). This is a $\Omega(k)$ gap.
\textbf{(Right)} A pivot $p$ opens at $0$. Points at $\nicefrac{1}{4}-\varepsilon$ are assigned to $p$, but later arrivals push the average distance above $\nicefrac{1}{4}$, forcing deactivation. The algorithm (top) splits tightly grouped points near $\nicefrac{1}{4}$ across two clusters, incurring massive cost for separated positive edges between clusters (solid blue edges); the optimum (bottom) keeps all points in a single cluster, paying only for $p$'s negative edges (solid red edges). This is a $\Omega(k)$ gap.
\label{lineexample}}
\end{figure}

This algorithm has a natural online adaptation. Maintain a set of ``active'' pivots, and
when $u$ arrives, admit it to a nearby active pivot's cluster if one exists
within the clustering radius. Otherwise, open a new active pivot at $u$. Monitor each active cluster's average distance to its pivot (informally, the pivot's \emph{density}), and
once that average becomes too low, \emph{deactivate} the cluster, i.e., admit no further nodes to the cluster.

Deactivating a cluster permanently due to the pivot having low density at some time leads to a large competitive ratio. Intuitively, this is because the density around a node can drastically change based on arrival order.  The right-hand panel of Figure~\ref{lineexample} illustrates how this online version of the algorithm by Charikar, Guruswami, and Wirth fails. First, node $p$ arrives at $0$ and becomes an active pivot. Then $k$ points arrive at
$\nicefrac14-\varepsilon$ and $2k$ points at $\nicefrac14+\varepsilon$; all lie within the
clustering radius of $p$ and join its cluster, but together they push the average
distance to $p$ above $\nicefrac14$, forcing the cluster to deactivate. Now $k^2$ points arrive
at $\nicefrac14$. With $p$'s cluster permanently deactivated, the newly arrived nodes form a new, separate cluster.
The resulting solution pays for the positive edges between the $k^2$ points at
$\nicefrac14$ and the $\Theta(k)$ earlier points clustered with $p$. Those pairs are
at distance $\varepsilon$, so their total cost is
$\Theta(k^3 \cdot (1-\varepsilon))$. A different solution is to cluster all nodes together, which pays for $p$'s
negative weight of at most $\nicefrac14+\varepsilon$ to each of the $\Theta(k^2)$ nodes that arrived later, for a
total cost of $\Theta(k^2)$.  The rest of the negative weight paid by this solution is at most $\Theta( \varepsilon \cdot k^3)$. Taking $\varepsilon \leq \nicefrac{1}{k},$ the cost of the former clustering is $\Theta(k^3)$, while the cost of the latter clustering is $\Theta(k^2)$.
Choosing
$k = \Theta(\sqrt{n})$ gives a $\Omega(\sqrt{n})$ lower bound on this approach.

The problem with the previous example is that we deactivated a cluster  because its mass
moved, but the mass moved to a concentrated location. In particular, deactivating the cluster in our example ignores that a
dense cluster has formed around $\nicefrac14$. Intuitively, the algorithm should instead
relocate the center from $0$ toward $\nicefrac14$ and keep admitting points.
However, relocation cannot be free, because if a center could jump anywhere, clusters could
drift far across the metric, again hurting the cost. 

Our strategy is to 
check a pivot's density, and move the pivot to a denser node if necessary and when one exists,  only when its cluster doubles in size. This ensures that we do not deactivate a cluster too early, while keeping
the distance (with respect to the semi-metric $d=w^-$) between consecutive centers of a
cluster small, even though the cumulative drift of the center over the cluster's lifetime
can be arbitrarily large.

\paragraph{Our algorithm}
We organize each cluster's life into phases, where a cluster is in phase $\ell$
while its size lies in $[2^\ell, 2^{\ell+1})$. An arriving $u$ joins the
cluster of an active pivot $p$ with $d_{up}\le r$ (with $r$ playing the role of CGW's
clustering radius $\nicefrac12$), opening a new, initially singleton cluster otherwise. Ties
are broken by a priority ordering $\succ$ on active clusters that (roughly speaking) favors the
larger cluster; we note this tie-breaking rule is important in the analysis and that arbitrary tie-breaking can result in a non-constant competitive ratio. We check whether the pivot of cluster $C$ should be moved versus whether $C$ should be (permanently) deactivated only when
$|C|$ is a power of $2$. To determine which decision to make, we check whether there is a node $q$ in $C$ that is
dense, meaning its average distance to its cluster is at most $\rho$ (with $\rho$ playing
the role of CGW's density threshold $\nicefrac14$, and $\rho < r$), that is, 
\[
\frac{1}{|C|}\sum_{v\in C} d_{vq} \;\le\; \rho .
\]
If such a $q$ exists, we move the pivot to $q$ and begin the next phase for that cluster; otherwise, we deactivate the cluster, closing it off to any new nodes permanently.

Although density is only tested at the start of each phase, because a cluster's size only doubles before the next test, a pivot dense at
the start of a phase remains a good representative throughout it. Given a pivot $p$, we call the
set of nodes in $p$'s cluster, at the time $p$ becomes a pivot, that are within distance $2\rho$ of $p$, the
\emph{dense ball} $T_p$. By Markov's inequality on the density condition,
$|T_p| \ge \tfrac{1}{2}|C|$ at the time when $p$ is confirmed as a pivot. Moreover,
since any new pivot $q$ must satisfy the density condition over the same cluster
as the outgoing pivot $p$, the triangle inequality forces $d_{pq} \le 3\rho$,
so pivots move \emph{gracefully}: consecutive pivots are always close. Algorithm~\ref{alg: main-alg} gives a formal description.

We preview some of the challenges in our analysis. Let $\mathcal{C}$ be
the clustering our algorithm outputs and $\mathcal{C}^*$ a fixed optimal
clustering. It is easy to charge the cost of edge $uv$ incurred by the
algorithm if (1) $\mathcal{C}$ and $\mathcal{C}^*$ agree on whether $u$ and
$v$ should be in the same cluster or (2) $d_{uv}$ is bounded away from $0$ and
$1$. The hard cases are when $d_{uv} \leq \varepsilon$ (for $\varepsilon =
\nicefrac{\rho}{4}$) and $\mathcal{C}$ separates $u$ and $v$ even though $\mathcal{C}^*$
keeps them together, or when $d_{uv} > 1-\varepsilon$ and $\mathcal{C}$
clusters $u$ with $v$ even though $\mathcal{C}^*$ separates them. In these
cases, we must charge the cost incurred on $uv$ to other edges. Consider an
edge $uv$ with $d_{uv} \leq \varepsilon$, where $\mathcal{C}$ separates $u$
and $v$, $\mathcal{C}^*$ keeps them together, and $u$ arrives after $v$. Since
$u$ did not join $v$'s cluster, one of three things must have happened:
\begin{itemize}
\item[(i)] $v$'s cluster was still active, but the current pivot of $v$'s
cluster was too far from $u$ to admit it (Section~\ref{sec:X});
\item[(ii)] $v$'s cluster was still active and its pivot was within radius $r$
of $u$, but a higher-priority active pivot was also close enough to claim $u$
(Section~\ref{sec:Y});
\item[(iii)] $v$'s cluster was deactivated before $u$ arrived
(Section~\ref{sec:Z}).
\end{itemize}
In most cases, we identify a pivot $p^*$, active between the arrivals of $v$
and $u$, at intermediate distance from $u$, and charge the cost of $uv$ to
edges $uz$ for $z \in T_{p^*}$. Since $d_{uv} \le \varepsilon$ and $v$ was
admitted at radius $r$, the pivot of $v$'s cluster starts at distance at most
$r + \varepsilon$ from $u$; in case~(i) it has since moved beyond $r$, so
graceful movement places some intermediate pivot $p^*$ in a window $(r, \beta]$
for a suitable constant $\beta < 1$. In case~(iii), deactivation
certifies that no node (including $v$) of the frozen cluster was dense, forcing a
constant fraction of the cluster to lie at distance $\Omega(\rho)$ from $v$ (and thus from $u$), then the
graceful movement of pivots lets us find some intermediate pivot(s) $p^*$ and charge to their dense ball(s) $T_{p^*}$. Case~(ii) is the most delicate: if $v$'s
cluster ever swung far from $u$ despite its current pivot returning within $r$,
the same intermediate-value argument finds $p^*$ inside $v$'s cluster;
otherwise $v$'s cluster stayed uniformly close, and the priority ordering
$\succ$ forces a pivot of $u$'s \emph{own} cluster into the intermediate
window, with the ordering preventing overcharging across multiple clusters
that may simultaneously lie close to $u$.

\section{Preliminaries}

\subsection{Notation}

In our analysis, we take $d_{uv}=w^-_{uv}$ for every edge $uv$. We find this to be less clunky, and indeed the choice of $d$ is because $w^-$ induces a semi-metric space on $V$, where we take $d_{uu}=0$ for all $u$. 

We let $\mathcal{C}$ denote a clustering (a collection of clusters), typically, our algorithm's clustering, and specially denote an optimal clustering as $\mathcal{C}^*$. We let $\binom{\mathcal{C}}{2}$ denote  the set of all (unordered) pairs of clusters $C,C'$ in $\mathcal{C}.$
For any clustering $\mathcal{C}$, we let $\mathcal{C}(u)$ for $u \in V$ denote the cluster of $\mathcal{C}$ containing $u$.
We say a cluster $C$ is in \emph{phase} $\ell$ if $2^\ell \leq |C| < 2^{\ell+1}$, and let $C_{ \ell}$ denote the $2^\ell$ nodes clustered into $C$ before $C$ entered phase $\ell$.

Our algorithm maintains a dynamic set $P$ of pivots, one for each active cluster (defined below). For a cluster $C$, the current \emph{pivot} $p \in C$ is the node that is used to decide whether newly arrived nodes will be added to $C$.
We interchangeably use the terms pivots and  cluster centers. The pivot for a cluster may move over time.
During any phase, a cluster $C$ either has exactly one pivot, in which case we say the cluster (or the corresponding phase) is \emph{active} and newly arrived nodes can be clustered into $C$; or the cluster is \emph{inactive}. The only times we change the pivot of $C$ or $C$'s active status is when $C$ starts a new phase, i.e., when $|C|$ is a power of 2.
Once a cluster becomes inactive, it stays inactive.

For a given cluster, we view the set of pivots over all of the cluster's active phases as an ordered multi-set; it contains a copy of each node for each phase in which that node is a pivot, and these are ordered in increasing order of phase. Moreover, every pivot $p$ maps uniquely to a cluster, namely $\mathcal{C}(p)$ in the notation above. Further, since we view the set of pivots as a multi-set, we may uniquely map every pivot to an active phase for that cluster, namely the phase for which it is the pivot. As such, it will be useful to let $\phi(p)$ denote the phase for which $p$ is the pivot for its cluster $\mathcal{C}(p)$. We say a node $p$ initiating a cluster $C$ has no pivot. However, $p$ then becomes the pivot for $C$ for phase 0. For all other vertices $v$ in $C$, $v$ is clustered by one of $C$'s pivots $p$, and we say in this case that $v$ is clustered during phase $\phi(p)$.

We track an ordering $\succ$ on active clusters, or equivalently, on current pivots. For $C$ and $C'$ clusters, we say $C \succ C'$ while $C$ is in phase $\ell$ if $C$ enters phase $\ell$ before $C'$ enters phase $\ell$. (In particular, if $C$ and $C'$ overlap in the times that they are in phase $\ell$, then the cluster that reached phase $\ell$ first is larger with respect to $\succ$.) At a time when $C$ is non-empty and $C'$ is empty, we may extend the ordering and write $C \succ C'$. Note that this ordering may change over time.

Throughout the analysis, for any vertex $u$ and $\delta \geq 0$, we call the set $\{v \mid d_{uv} \leq \delta\}$ the  \emph{$\delta$-ball of $u$}, where the ground set for $v$ will be clear from context, and is often restricted to a fixed cluster.
In our analysis, a particularly important set of balls to consider are those around pivots, with fixed radius $2 \rho$, for $\rho >0$ a parameter of Algorithm \ref{alg: main-alg}.
Specifically, for a pivot $p$, we define
\[T_p = \{ z \in \mathcal{C}(p) \mid d_{zp} \leq 2 \rho, \; z \text{ is clustered before phase }\phi(p)\},\]
and often call $T_p$ the \emph{dense ball around $p$}. Also, we say node $q$ has \emph{average density} at most $\rho$ to a set of nodes $S$ when $\frac{1}{|S|}\cdot \sum_{v \in S} d_{vq} \leq \rho$.

\subsection{Algorithm}

We state our algorithm formally in Algorithm \ref{alg: main-alg}. Note we write the ordering $\succ$ to be on the current set $P$ in the algorithm, and this induces the analogous ordering on the corresponding active clusters.

\begin{figure}[h]
  \makebox[\linewidth]{%
  \scalebox{0.9}{\begin{minipage}{\dimexpr\linewidth-7em}
\begin{algorithm}[H]
\setstretch{.7}
\caption{Main Algorithm}\label{alg: main-alg}
\begin{algorithmic}[1]
    \State \textbf{Input: }  $G=(V, E)$ with weights $d$ satisfying triangle inequality, constants $r = \nicefrac{18}{115}$ and $\rho = \nicefrac{r}{9} = \nicefrac{2}{115}$
    \State  \textbf{Initialize: }  empty clusters $\mathcal{C}=\emptyset$ and pivots $P = \emptyset$ with an ordering $\succ$ on $P$
    \State Set $d_{uv}=w^-_{uv}$ for all $uv \in E$ and $d_{uu}=0$
\For{  each arriving $u \in V$}
 \If{ there exists a pivot $p \in P$ with $d_{up} \leq r$}
    \State Add $u$ to $p$'s cluster $\mathcal{C}(p)$ (break ties by choosing the highest such $p$ w.r.t. $\succ$)
  \If{$|\mathcal{C}(p)| = 2^{\phi(p)+1}$}
  \If{ $\exists q \in \mathcal{C}(p)$ with  $\frac{1}{|\mathcal{C}(p)|}\cdot \sum_{v \in \mathcal{C}(p)} d_{vq} \leq \rho$} \label{eq: avg-density-pseudocode}
  \State Update $P \leftarrow P \setminus \{p\} \cup \{q\}$ \Comment{Identify $\mathcal{C}(q)$ with $\mathcal{C}(p)$}
  \State Set $\phi(q) = \phi(p) +1$
  \State Update $\succ$ on $P$ so $q \succ p'$ for all $p' \in P$ with $\phi(p') < \phi(q)$,  and $p'' \succ q$ \indent \indent \indent for every $p'' \in P$ with $\phi(p'') \geq \phi(q)$
  \Else
 \State Update $P \leftarrow P \setminus \{p\}$ \Comment{$\mathcal{C}(p)$ is now inactive.}
  \EndIf
  \EndIf
      \Else
    \State Add $u$ to $P$, 
    \State Add $\{u\}$ to $\mathcal{C}$, set $\mathcal{C}(u) = \{u\}$ 
    \State Set $\phi(u)=0$ 
  \EndIf
    \EndFor
    \State  \textbf{Output: } Clustering $\mathcal{C}$
\end{algorithmic}
 \end{algorithm}
    \end{minipage}}}
\end{figure}

 \section{Analysis of Competitive Ratio}

In this section, we prove Theorem \ref{thm:main}.

We have to charge the cost of all of the algorithm's mistakes to the cost of the optimal solution. Fix some optimal clustering $\mathcal{C}^*$, with cost 
\[\text{OPT} = \sum_{C \in \mathcal{C}^*} \sum_{u,v \in C}d_{uv} + \sum_{C,C' \in \binom{\mathcal{C}^*}{2}}\sum_{u \in C} \sum_{v \in C'}(1-d_{uv}).\]
We also denote the cost of $\mathcal{C}^*$ adjacent to $u$ as $\text{OPT}(u)$, where \[\text{OPT}(u)=\sum_{v \in \mathcal{C}^*(u)}d_{uv} + \sum_{C' \neq \mathcal{C}^*(u)} \sum_{v \in C'}(1-d_{uv}),\]
and note $\text{OPT} = \frac12 \cdot \sum_u \text{OPT}(u).$
For $\mathcal{C}$ the output of Algorithm \ref{alg: main-alg}, its cost is
\begin{align}
    \text{ALG} = \underbrace{\sum_{C \in \mathcal{C}} \sum_{uv \in C}d_{uv} 
}_{S^-}+ \underbrace{\sum_{C,C' \in \binom{\mathcal{C}}{2}}\sum_{u \in C} \sum_{v \in C'}(1-d_{uv})}_{S^+}\label{eqn: plus-minus-part}
\end{align}
We will bound $S^+$ in Subsection \ref{subsec:pos} and $S^-$ in Subsection \ref{subsec:neg}. Bounding $S^+$ takes substantially more work than bounding $S^-$; this is because there are different  algorithmic reasons why nodes $u,v$ may be clustered separately, and each reason requires a different charging argument (see \ref{sec:X}, \ref{sec:Y}, and \ref{sec:Z}).

We begin with a few statements that will be used in bounding $S^+$ and $S^-$.
The following fact on average density unifies the case of when a node is a pivot because it initiates a cluster, versus when a node becomes a pivot after the cluster was already initiated. Recall $C_{\phi(p)}$ are the nodes in $C$ when $C$ enters phase $\phi(p)$.

\begin{fact}\label{fact: dense-pivot}
    For $p$ a pivot of $C$, the average density of $p$ to $C_{\phi(p)}$ is at most $\rho$.
\end{fact}
\begin{proof}
When $\phi(p)=0$, we have $C_{\phi(p)}= \{p\}$, so the average density of $p$ to every node in $C_{\phi(p)}$ is $d_{pp} =0$.
Now we consider $p$ with $\phi(p)>0$. The average density of $p$ to $C_{\phi(p)}$ is at most $\rho$ by choice of $p$, specifically line~\ref{eq: avg-density-pseudocode} of Algorithm \ref{alg: main-alg}.
\end{proof}

We use Fact \ref{fact: dense-pivot} in proving the following two claims.
The first claim states that a constant fraction of a cluster is concentrated around its pivot.

\begin{claim}\label{claim: dense-pivot}
    Let $p$ be one of the pivots for cluster $C$. Then $p$ has at least $ 2^{\phi(p)-1}$ nodes from $C$ in its $2\rho$-ball when $C$ enters phase $\phi(p)$ (and thereafter). In other words, for 
    \[T_p = \{z \in \mathcal{C}(p) \mid d_{zp} \leq 2 \rho,\; z \text{ is clustered before phase }\phi(p)\},\] it holds that $|T_p| \geq 2^{\phi(p)-1}$. 
\end{claim}
\begin{proof}[Proof of Claim \ref{claim: dense-pivot}]
 Since $p$ is a pivot for $C$, by Fact \ref{fact: dense-pivot} $p$ has average density at most $\rho$ to $C_{\phi(p)}$ when phase $\phi(p)$ begins. By Markov's inequality, at least half of $C_{\phi(p)}$ is within distance $2  \rho$ to $p$ (as if this were not true, then the sum of distances between $C_{\phi(p)}$ and $p$ is more than $\frac{1}{2} \cdot |C_{\phi(p)}|\cdot 2 \rho =   2^{\phi(p)}\cdot \rho$, breaking the assumption that the average distance is at most $\rho$).
\end{proof}

The next claim states that consecutive pivots of a cluster (the pivots for consecutive phases) are relatively close to each other. We colloquially refer to this useful property throughout the paper as \emph{graceful movement}. This will be very important in our analysis, as it gives us local control over charging a cluster's cost.

\begin{claim}
\label{claim: pivot-migration}
Let $C$ be a cluster with consecutive pivots $p$ and $q$, that is, $\phi(q) = \phi(p)+1$. Then $d_{pq} \leq 3 \rho$.
\end{claim}
\begin{proof}[Proof of Claim \ref{claim: pivot-migration}]
   Since $p$ is a pivot for $C$, we have by Fact \ref{fact: dense-pivot} that $ \sum_{v \in C_{\phi(p)}} d_{vp} \leq \rho \cdot 2^{\phi(p)}$.
   We examine the distance of $q$ to the nodes in $C_{\phi(p)}$ versus $C_{\phi(q)}$:
       \begin{align*}
       \rho \cdot 2^{\phi(p)+1} = \rho \cdot 2^{\phi(q)} &\geq \sum_{v \in C_{\phi(p)}} d_{vq} + \sum_{z \in C_{\phi(q)}\setminus C_{\phi(p)}} d_{zq}   \\
        &\geq \sum_{v \in C_{\phi(p)}} (d_{pq}-d_{vp})+ 0 \cdot 2^{\phi(p)}\\
        &= 2^{\phi(p)}  \cdot d_{pq} - \sum_{v \in C_{\phi(p)}} d_{vp}\\
        &\geq 2^{\phi(p)}  \cdot (d_{pq} - \rho).
    \end{align*}
    Note the second and third inequality use that $|C_{\phi(q)}\setminus C_{\phi(p)}|=|C_{\phi(p)}| = 2^{\phi(p)}$.
    Dividing both sides of the above inequality by $2^{\phi(p)}$ and rearranging, we have $3  \rho \geq d_{pq}.$ 

\end{proof}

\subsection{Bounding $S^+$}\label{subsec:pos}

We further partition $S^+$, defined in Equation (\ref{eqn: plus-minus-part}),
into 
\[S^+ = \sum_{uv \in A_1}(1-d_{uv})+\sum_{uv \in A_2}(1-w^-_{uv})+\sum_{uv \in A_3}(1-d_{uv})\]
with 
\begin{align*}
    A_1 &= \{uv \mid \mathcal{C}(u) \neq \mathcal{C}(v),~ \mathcal{C}^*(u) \neq \mathcal{C}^*(v)\},\\
    A_2 &= \{uv \mid \mathcal{C}(u) \neq \mathcal{C}(v),~ \mathcal{C}^*(u) = \mathcal{C}^*(v),~ d_{uv} \geq \varepsilon \},\\ 
\text{and } A_3&= \{uv \mid \mathcal{C}(u) \neq \mathcal{C}(v),~ \mathcal{C}^*(u) = \mathcal{C}^*(v),~ d_{uv} < \varepsilon \},
\end{align*}
where we take $\varepsilon = \nicefrac{\rho}{4}$. 
In words, $A_1$ are the $uv$ pairs where our clustering $\mathcal{C}$ and the optimal clustering $\mathcal{C}^*$ agree in separating $u$ and $v$, whereas $A_2$ and $A_3$ contain $uv$ pairs where $\mathcal{C}$ and $\mathcal{C}^*$  disagree on whether or not to cluster $u$ and $v$ together or separate. In $A_2$, however, we can easily charge to $d$ since we separate pairs whose distance is not too small, whereas in $A_3$ we need to charge to other edges. 

We charge the cost of $uv \in A_1$ and $uv \in A_2$ in the next lemma.

\begin{lemma}\label{lem: A1A2}
    \[\sum_{uv \in A_1}(1-d_{uv})+\sum_{uv \in A_2}(1-d_{uv})\leq \nicefrac{1}{\varepsilon}\cdot \text{OPT}.\]
\end{lemma}
\begin{proof}[Proof of Lemma \ref{lem: A1A2}]
  For $uv \in A_1$, both $\mathcal{C}$ and $\mathcal{C}^*$ pay the same cost on that edge.
For $uv \in A_2$, $\mathcal{C}$ pays $1-d_{uv} \leq 1$, while $\mathcal{C}^*$ pays at least $\varepsilon$, so $1-d_{uv} \leq 1\leq \nicefrac{1}{\varepsilon} \cdot d_{uv}$.
Together we bound $uv \in A_1 \cup A_2$ with
\[\sum_{uv \in A_1}(1-d_{uv})+\sum_{uv \in A_2}(1-d_{uv}) \leq \sum_{C,C' \in \binom{\mathcal{C}^*}{2}}\sum_{u \in C} \sum_{v \in C'}(1-d_{uv})+ \sum_{uv \in A_2}\nicefrac{1}{\varepsilon} \cdot d_{uv} \leq \nicefrac{1}{\varepsilon}\cdot \text{OPT}.\] 
\end{proof}

We continue by bounding the cost of edges in $A_3$.
 For $uv \in A_3$, let $u$ and $v$ be labeled so that $v$ arrived before $u$.
 Since $u$ and $v$ are not in the same cluster, one of the following occurred when $u$ arrived: 
\begin{enumerate}
    \item[$X$:] The cluster containing $v$ was still active, and $u$ was too far away from the current pivot of $v$'s cluster.
    \item[$Y$:] The cluster containing $v$ was still active, but $u$ was not too far from the current pivot of $v$'s cluster. This means there was a pivot higher in the current ordering $\succ$ on $P$ that was within $u$'s $r$-ball.
    \item[$Z$:] The cluster containing $v$ was inactive.
\end{enumerate}

Partition the $uv \in A_3$ based on which of the 3 cases above it falls into, call these sets $X,Y,Z$. Throughout, when we write $uv \in X$ (likewise for $Y, Z$) we assume the pair is labeled so $v$ arrived before $u$.

\subsubsection{Charging $uv \in X$}\label{sec:X}

For $uv \in X$ with $v \in C$, $u$ is not clustered with $v$ in cluster $C$ because $C$'s pivot when $u$ arrived was too far from $u$. The pivot of $C$ when $v$ arrived was close to $u$ though (or if $v$ initiated $C$, we know $v$ is close to $u$), 
so by the graceful movement of consecutive pivots (Claim \ref{claim: pivot-migration}),
there is some pivot of $C$ with intermediate distance from $u$.
The dense ball around this pivot will be used to charge such $uv$'s cost.

\begin{figure}
    \centering
    \includegraphics[width=0.7\linewidth]{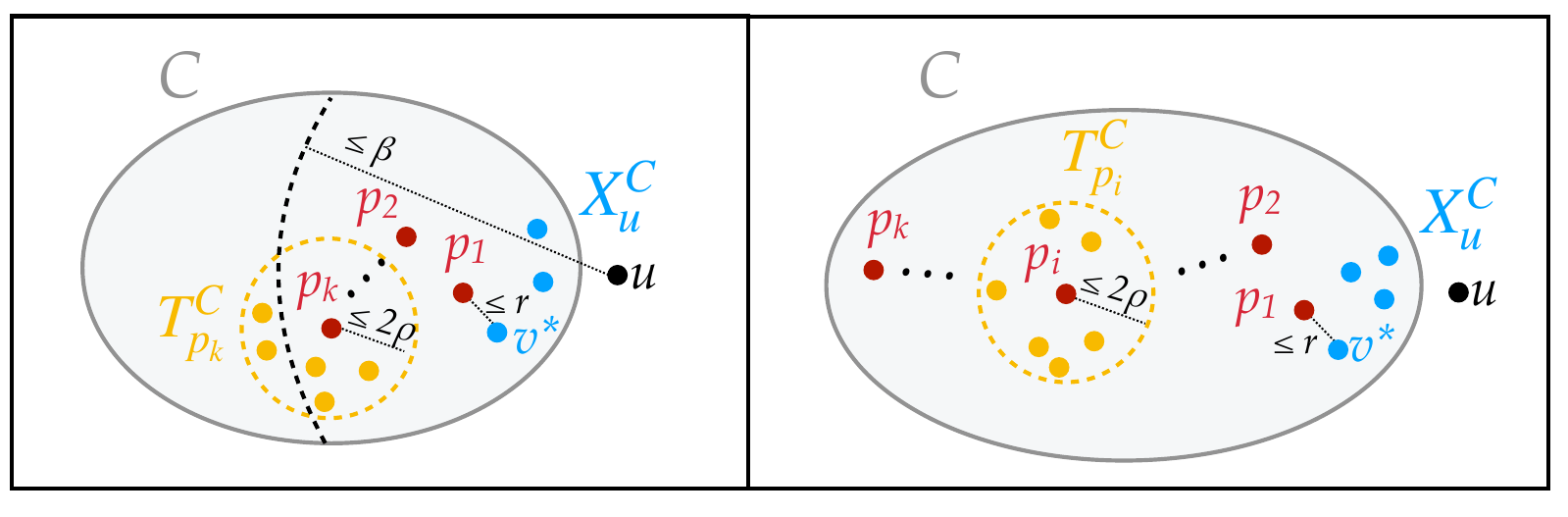}
    \caption{Accompanies the proof of Lemma \ref{lem:x-charge}. \textbf{(Left)} Illustrates Case 1, where $d_{up_k} \leq \beta$.  Edges $uv$ for $v \in X_u^C$ (light blue) are charged to edges $uz$ for $z \in T_{p^*_C}$ (yellow), where $p^*_C=p_k$.  \textbf{(Right)} Illustrates Case 2, where $d_{up_k} > \beta$. Edges $uv$ for $v \in X_u^C$ (light blue) are charged to edges $uz$ for $z \in T_{p^*_C}$ (yellow), where $p^*_C=p_i$ for a carefully chosen $p_i$.\\
    In both cases, pivots of $C$ are red, and $p^*_C$ is chosen to ensure that $u$ is not too far from but also not too close to $p^*_C$.}
    \label{fig:charge-pos-x}
\end{figure}

\begin{lemma}{\label{lem:x-charge}}
    \[\sum_{uv \in X}(1-d_{uv}) \leq   8 \cdot \max \left \{\nicefrac{1}{1-14\rho},\nicefrac{1}{7\rho} \right \}\cdot \text{OPT} .\]
\end{lemma}
\begin{proof}[Proof of Lemma \ref{lem:x-charge}]

Fix $u$, and fix a cluster $C$ not containing $u$. Let $X^C_u = \{v \in C \mid uv \in X, v \text{ arrives before }u\}$. Let $v^* \in X_u^C$ be the last node to arrive in $X_u^C$.

Consider the pivots of $C$ between when $v^*$ arrives and $u$ arrives (inclusive), and label these pivots with the ordered multi-set $P_u^C= \{p_1,\ldots,p_k\}$. Specifically, if $v^*$ initiates $C$, $v^*$ automatically becomes a pivot, so we have $p_1=v^*$, and if $v^*$ does not initiate $C$, then $p_1$ is the pivot of $C$ when $v^*$ arrives (which must exist in this case). Either way, we find $d_{up_1} \leq \varepsilon+r$: if $v^*$ initiates $C$ then $d_{up_1} = d_{uv^*} \leq \varepsilon$, and if $v^*$ does not initiate $C$, then $d_{up_1} \leq d_{v^*p_1}+d_{uv^*} \leq r+\varepsilon$.
Further, $p_k$ is the pivot of $C$ when $u$ arrives, so $d_{up_k} >r$ by the definition of $X$. Note that when $v^*$ initiates $C$, and therefore $p_1=v^*$, $d_{up_k} >r$ and $d_{up_1}\leq \varepsilon$. Since $\varepsilon<r$, it must be that $p_k \neq p_1$ in this scenario.
\smallskip

We case on whether $d_{up_k}$ is large or small. Define $\beta := r+3\rho$. This is a local parameter used only for shorthand in this proof. We show in both cases that there is a pivot $p_C^*$ such that $|T_{p_C^*}| = \Omega(|X_u^C|)$ and $d_{uz}, 1-d_{uz} = \Omega(1)$ for each $z \in T_{p_C^*}$. This will allow us to charge the cost of $X_u^C$ to $\text{OPT}(u)$.

\setcounter{case}{0}
\begin{case}
$d_{up_k} \leq \beta$.
\end{case}
Here, we can charge our incurred cost of $1-d_{uv} \leq 1$ to $d_{uz}  \geq d_{up_k} - d_{p_kz} >   r-2\rho$ or $1-d_{uz} \geq 1-\beta-2 \rho$, for $z \in T_{p_k}$, where recall $ T_{p_k}$ is the $2\rho$-ball around $p_k$ that were clustered by $C$ before phase $\phi(p_k)$. By Claim \ref{claim: dense-pivot}, $|T_{p_k}| \geq 2^{\phi(p_k)-1} \geq 2^{\phi(p_1)-1}$. Since all of $X_u^C$ arrived during or before phase $\phi(p_1)$, we have $|X_u^C| \leq 2^{\phi(p_1)+1}\leq  4 \cdot  |T_{p_k}|$. Thus in this case, we can take $p^*_C=p_k$.

\begin{case}
$d_{upk} > \beta$.
\end{case}
We claim 
there is some pivot $p_i \in P_u^C$ where $r<d_{up_i}\leq \beta$. 
In particular, let $i$ be the maximum index such that $ d_{up_i} \leq \beta$. (To see that such $i$ exists, we note that $d_{up_1} \leq r+\varepsilon \leq \beta$.) Further, note that $i \neq k$, since $d_{up_k} > \beta$ by assumption. We claim that $d_{up_i} > r$. If this were not true, then $d_{up_{i+1}} \leq d_{up_i} + d_{p_ip_{i+1}} \leq r + 3\rho \leq \beta$, which contradicts that $i$ was chosen maximally.
Taking such a pivot $p_i$, and $z \in T_{p_i}$, we see $1-d_{uz} \geq 1-(d_{up_i} + d_{p_iz} )\geq 1-\beta-2\rho$ and $d_{uz} \geq d_{up_i}-d_{zp_i} \geq r-2\rho$. As before, $|X_u^C| \leq 4 \cdot  |T_{p_i}|$. It follows that in this case we can take  $p^*_C=p_i$.

\medskip

Combining both cases, we see that, for every $C \not \ni u$, we have identified a pivot $p^*_C$ such that $u$'s distance to every node in $T_{p^*_C}$ is bounded above by $\beta+2\rho$ and below by $r-2\rho$. Moreover, $|X_u^C| \leq 4 \cdot |T_{{p^*_C}}|$. So in total, for fixed $u$, we have
\begin{align*}
    \sum_{C \not \ni u}  \sum_{v \in X_u^C}(1-d_{uv})
    & \leq   \sum_{C \not \ni u} |X_u^C|\\
    & \leq  4\cdot \sum_{C \not \ni u} |T_{p^*_C}| \\
    & \leq  4 \cdot  \sum_{C \not \ni u}
      \Bigg (\sum_{\substack{z \in T_{p^*_C}: \\ \mathcal{C}^*(u)= \mathcal{C}^*(z)}} \nicefrac{1}{(r-2\rho)}\cdot d_{uz}+
    \sum_{\substack{z \in T_{p^*_C}:\\ \mathcal{C}^*(u) \neq \mathcal{C}^*(z)}} \nicefrac{1}{(1-\beta-2\rho)}\cdot (1-d_{uz})  \Bigg )\\
    & \leq  4 \cdot \max \left \{\nicefrac{1}{(r-2\rho)}, \nicefrac{1}{(1-\beta-2\rho)}, \right \}\cdot \text{OPT}(u).
\end{align*}
Then we can sum over all $u$ to see
\begin{align*}
    \sum_{uv \in X}(1-d_{uv}) &= \sum_{u} \sum_{C \not \ni u} \sum_{v \in X_u^C}(1-d_{uv})\\
    &\leq  4 \cdot \max \left \{\nicefrac{1}{(r-2\rho)}, \nicefrac{1}{(1-\beta-2\rho)}, \right \}\cdot  \sum_{u}\text{OPT}(u) \\
    &=  8 \cdot \max \left \{\nicefrac{1}{(r-2\rho)}, \nicefrac{1}{(1-\beta-2\rho)}, \right \}\cdot \text{OPT}.
\end{align*}
We note that every $u$ only charges $uv$ with $v \in X_u^C$ to $uz$ for $z \in T_{p^*}^C$. Since $|X_u^C| \leq 4 \cdot |T_{{p^*_C}}|,$ each edge is charged at most 4 times, and this is accounted for in the bound. No other overcharging occurs since the charging argument is local between $u$ and $C$.

\end{proof}

\subsubsection{Charging $uv \in Y$}\label{sec:Y}
For $uv \in Y$ where $u$ arrives after $v$ and $v \in C$, $u$ is not clustered with $v$ in cluster $C$ because there was a different active cluster, $C'$, whose pivot was close enough to $u$ to cluster it and $C' \succ C$ at that time.  Our main lemma for this subsection is the following.

\begin{lemma}{\label{lem:y-charge}}
    \[\sum_{uv \in Y}(1-d_{uv}) \leq  8 \cdot \max\{\nicefrac{1}{\rho}, \nicefrac{1}{(1-11\rho - \varepsilon)}\} \cdot \text{OPT} + 16 \cdot \max\{\nicefrac{1}{(\rho-\varepsilon)}, \nicefrac{1}{(1-14\rho-\varepsilon)}\} \cdot \text{OPT} . \]
\end{lemma}

We use several helper lemmas to prove Lemma \ref{lem:y-charge}. We begin by defining some preliminaries relevant to all of these lemmas, and their combination.

Analogous to the previous subsection, fix $u$ and $C \not \ni u$. 
Then define $Y_u^C := \{v \in C \mid uv \in Y, v \text{ arrives before }u\}$, with $v^*$ the last vertex to arrive in $Y_u^C$. Further, let $P_u^C= \{p_1,\ldots,p_k\}$ be the set of pivots for cluster $C$ between when $v^*$ arrives and  when $u$ arrives, inclusive. When $v^*$ initiates  $C$, we take $p_1=v^*$, and otherwise $p_1$ is the pivot that clusters $v^*$.

The charging argument will depend on whether or not $u$ is far from any pivot in $P_u^C$, more formally, we characterize $C \not \ni u$ as being either \emph{close to} or \emph{far from} $u$.

\begin{definition}
       Fix $u \in V$ and fix a cluster $C$ not containing $u$. 
    We call a cluster $C$ that does \emph{not} contain $u$ \emph{close to $u$}, denoted $u \leadsto C$, if all 
    $p \in P_u^C$ have $d_{up} \leq 3 \rho$. Otherwise we call $C$ \emph{far from $u$}, denoted $u \not\leadsto C$. 
\end{definition}

    We first charge $uv$ for nodes $v \in C$ where $C$ is a cluster far from $u$; this case is much easier than when $v$ belongs to a cluster close to $u$. Here, we will use the fact that $p_1$ is close to $u$, but some $p \in P_u^C$ is sufficiently far from $u$. Thus by the graceful movement of pivots (Claim \ref{claim: pivot-migration}), there is some pivot at an intermediate distance, and we can charge to $uz$ for $z$ in that pivot's dense ball (Claim \ref{claim: dense-pivot}).

    \begin{lemma}\label{lem: caseY-C-far-from-u}
        \[\sum_{\substack{uv \in Y:\\ u \not\leadsto \mathcal{C}(v) }}(1-d_{uv}) \leq 8 \cdot \max\{\nicefrac{1}{\rho}, \nicefrac{1}{(1-11\rho - \varepsilon)}\} \cdot \text{OPT}.\]
    \end{lemma}

    \begin{proof}[Proof of Lemma \ref{lem: caseY-C-far-from-u}]

    We show that there is always a pivot $p_C^*$ such that $|T_{p_C^*}| = \Omega(|Y_u^C|)$ and $d_{uz}, 1-d_{uz} = \Omega(1)$ for each $z \in T_{p_C^*}$. This will allow us to charge the cost of $Y_u^C$ to $\text{OPT}(u)$. In particular, we take $p_C^*$ to be the pivot $p_i$  in $P_u^C$ of minimum index such that $d_{up_i} > 3 \rho$ (which exists since $C$ is far from $u$).
    
     \setcounter{case}{0}
    \begin{case} If $i=1$, then $3\rho < d_{up_i} \leq r+\varepsilon$.
    \end{case}

Note this case can only occur when $v^* \neq p_1$ (since $d_{uv^*} \leq \varepsilon = \nicefrac{\rho}{4}$), in which case $p_1$ is the pivot of $v^*$.
    The lower bound is by choice of $i$, and the upper bound follows from $d_{up_1} \leq d_{p_1v^*} + d_{uv^*} \leq r+\varepsilon$, where we have used that $v^*$ is clustered by $p_1$, and that $v^* \in Y_u^C$. 

    \begin{case} If $i > 1$, then $3\rho < d_{up_i} \leq 6\rho$.
    \end{case}

    The lower bound is by the choice of $i$, and the upper bound follows from $d_{up_i} \leq d_{up_{i-1}} + d_{p_{i-1}, p_i} \leq 3\rho + 3\rho$ where we have used the minimality of $i$ and that consecutive pivots are at most $3\rho$ apart (Claim \ref{claim: pivot-migration}).    

    \medskip

    All of $Y_u^C$ arrived before or while $p_C^*$ was active, so $|Y_u^C| \leq 2 \cdot |C_{\phi(p_C^*)}| = 2 \cdot 2^{\phi(p_C^*)} \leq 4 \cdot |T_{p_C^*}|$, where the last inequality follows from Claim \ref{claim: dense-pivot}.
    Further, using the bounds in the two cases above, we have that for all $z \in T_{p^*_C}$, $d_{uz} \geq d_{up^*_C} - d_{zp^*_C} \geq 3 \rho - 2 \rho = \rho$ and $d_{uz} \leq d_{up^*_C}+d_{zp^*_C} \leq r+\varepsilon +2 \rho = 11\rho +\varepsilon$.

For a fixed $u$, we see     
    \begin{align*}
    \sum_{\substack{C \in \mathcal{C}:\\ u \not\leadsto C}} \sum_{v \in Y_u^C} (1-d_{uv})
        & \leq  \sum_{\substack{C \in \mathcal{C}:\\ u \not\leadsto C}} |Y_u^C|\\
        &\leq  \sum_{\substack{C \in \mathcal{C}:\\ u \not\leadsto C}} 4 \cdot |T_{p^*_C}|\\
        &=  4 \cdot \sum_{\substack{C \in \mathcal{C}:\\ u \not\leadsto C}}  \Bigg (\sum_{\substack{z \in T_{p^*_C}: \\ \mathcal{C}^*(u) = \mathcal{C}^*(z)}}\nicefrac{1}{\rho}\cdot d_{uz}+\sum_{\substack{z \in T_{p^*_C}:\\ \mathcal{C}^*(u) \neq \mathcal{C}^*(z) }}\nicefrac{1}{(1-11\rho - \varepsilon)}\cdot (1-d_{uz}) \Bigg )\\
        & \leq 4 \cdot \max\{\nicefrac{1}{\rho}, \nicefrac{1}{(1-11\rho - \varepsilon)}\} \cdot \text{OPT}(u).
    \end{align*}
    where note we use that $T_{p_C^*}$ is contained in $C$, so no $uz$ appears in the inner sums for multiple $C$ and thus the last inequality follows. Then summing over all $u$, 
    \begin{align*}
    \sum_{\substack{uv \in Y:\\ u \not\leadsto\mathcal{C}(v) }} (1-d_{uv}) &= \sum_u \sum_{\substack{C \in \mathcal{C}:\\ u \not\leadsto C}} \sum_{v \in Y_u^C} (1-d_{uv}) \leq 8 \cdot \max\{\nicefrac{1}{\rho}, \nicefrac{1}{(1-11\rho -\varepsilon)}\} \cdot \text{OPT}.
    \end{align*}
    \end{proof}

Next, we consider how to bound the cost of $uv \in Y$ for $v $ in cluster $C$ where $C$ is close to $u$. Note the difficulty here is that if we wish to mimic the arguments used thus far, we would need to identify a pivot of $C$ between when $v^*$ and $u$ arrive, where the pivot is not too close and not too far from $u$. Since $C$ is close to $u$, \emph{no such pivot exists in $C$}. Therefore, we will actually look for dense pivots in clusters other than $C$. 

For $C'$ the cluster of $u$, the relationship between the pivots of $C$ and the pivots of $C'$ will be important. The next claim helps codify this, and shows that if clusters ever swap their order with respect to $\succ$, then their pivots in the phase immediately following the swap must be sufficiently far apart.

\begin{claim}\label{claim: far-swap}
Let $C$ and $C'$ be clusters and let $p$ be a pivot for $C$ at a time when $C \succ C'$. Suppose $C'$ reaches phase $\phi(p) + 1$ before $C$ reaches phase $\phi(p)+1$, and let $p'$ be the pivot for $C'$ when $C'$ is in phase $\phi(p)+1$. Then $d_{pp'} \geq 7\rho$. 
    \end{claim}
    \begin{proof}[Proof of Claim \ref{claim: far-swap}] 
        By assumption, $C$ enters phase $\phi(p)$ (i.e., has cluster size $2^{\phi(p)}$) before $C'$, but $C'$ enters phase $\phi(p)+1$ (i.e., has cluster size $2^{\phi(p)+1}$) before $C$. This means the $2^{\phi(p)}$ nodes in $C_{\phi(p)+1}' \setminus C_{\phi(p)}'$ arrive and are clustered in $C'$ when $C'$ is in phase $\phi(p)$, which occurs while $p$ is a pivot for $C$.\footnote{Note we are using the convention from Algorithm \ref{alg: main-alg} that if two clusters are in the same phase at some point in time, the cluster that reached that phase first is larger with respect to $\succ$.} Since $C \succ C'$ during this time, it must be the case that every $v \in C_{\phi(p)+1}' \setminus C_{\phi(p)}'$ has $d_{vp} >r$.

        Consider $C'$ at the time that it enters phase $\phi(p)+1$, that is, the set $C_{\phi(p)+1}' = C_{\phi(p')}'$. Using that $p'$ is chosen as the pivot when $C'$ enters phase $\phi(p)+1$, we thus have
        \[\rho \cdot 2^{\phi(p)+1} \geq \sum_{v \in C_{\phi(p)+1}'} d_{vp'} \geq \sum_{v \in C'_{\phi(p)+1} \setminus C'_{\phi(p)}} (d_{vp}-d_{pp'}) \geq 2^{\phi(p)} \cdot r - 2^{\phi(p)} \cdot d_{pp'}.\]
        Rearranging the left-hand and right-hand sides, we see that $d_{pp'} \geq r -2\rho = 7\rho$.      
    \end{proof}

    The next claim is helpful for finding pivots of intermediate distance to $u$.

    \begin{claim}\label{claim: walking-trick}
        Fix a cluster $C$ and a node $u$ that is clustered by $C$ while it is in phase $\ell$.  Further, suppose there is some phase $\ell' < \ell$ for $C$ with pivot $q_{\ell'}$ such that $d_{uq_{\ell'}}\geq 3 \cdot \rho$. Then there exists a phase $\ell^*$ for $C$, $\ell' \leq \ell^* < \ell$, with pivot $q_{\ell^*}$ such that $3 \rho \leq d_{uq_{\ell^*}} \leq 12\rho$.  Moreover, for all $z \in T_{q_{\ell^*}}$, $d_{uz} \geq \rho$ and $1-d_{uz} \geq 1-14\rho$. 
    \end{claim}
    \begin{proof}[Proof of Claim \ref{claim: walking-trick}]
In the proof of this claim, we use subscripts on pivots to denote their phase. Note that $\ell>0$ due to the claim's hypothesis that a phase $\ell' < \ell$ exists.
    
        Take $\ell^*$ to be the minimum among phases $\ell' \leq \ell^* < \ell$ such that $d_{uq_{\ell^*}} \leq r+3\rho$. To see that such a $q_{\ell^*}$ exists, we note that for $q_{\ell}$ the pivot of phase $\ell$, and $q_{\ell-1}$ the pivot of phase $\ell-1$, we have  $d_{uq_\ell} \leq r$ since $q_\ell$ clusters $u$, and $d_{q_{\ell}q_{\ell-1}} \leq 3 \rho$ by Claim \ref{claim: pivot-migration}. Therefore, $d_{uq_{\ell-1}} \leq d_{q_{\ell}q_{\ell-1}}+ d_{uq_\ell}\leq3\rho + r$. So $\ell^* \leq \ell-1$ exists. 
        
        We claim that $d_{uq_{\ell^*}} \geq 3 \rho$. 
        Suppose not (in which case it must be that $\ell^* > \ell'$, since we already know $d_{uq_{\ell'}} \geq 3 \rho$), and let the pivot $q_{\ell^*-1}$ be the pivot of $C$ in phase $\ell^*-1$. 
        Here, $d_{uq_{\ell^*-1}} \leq d_{uq_{\ell^*}} + d_{q_{\ell^*}q_{\ell^*-1}} \leq 3\rho + 3\rho <r + 3\rho$ (where we use Claim \ref{claim: pivot-migration} to see consecutive pivots have bounded distance), which contradicts that $\ell^*$ was chosen minimally. 
Taking such a pivot $q_{\ell^*}$, and $z \in T_{q_{\ell^*}}$, we see $1-d_{uz} \geq 1-(d_{uq_{\ell^*}} + d_{q_{\ell^*}z} )\geq 1-r-3\rho-2\rho = 1-14\rho$ and $d_{uz} \geq d_{uq_{\ell^*}}-d_{zq_{\ell^*}} \geq 3\rho-2\rho=\rho$.
    \end{proof}

We begin to tackle charging $uv \in Y$ where $v \in C$ and $C$ is a cluster close to $u$. Let $C'$ be the cluster of $u$. The next lemma (Lemma \ref{lem:single-cluster}) proves that for $C$ a cluster close to $u$, the disagreements $uv$ for $v \in Y_u^C$ can be charged to $uz$ for $z \in T_p$, where $p$ is a pivot of $C'$ that has intermediate distance to $u$. See Figure \ref{fig:y-case}. As previously stated, this differs from the charging arguments in Lemmas \ref{lem:x-charge} and \ref{lem: caseY-C-far-from-u}, in which $uv$ is charged to $uz$ for $z$ a dense ball around a pivot of cluster $C$. The charging here will crucially use that Algorithm \ref{alg: main-alg} orders the clusters based on their size, and uses this ordering as a tie-breaking rule for assigning clusters. 

Importantly, Lemma \ref{lem:single-cluster} is \emph{not} sufficient on its own to charge the disagreements over all $C$ close to $u$, as many close $C$ all charging to dense balls around pivots in $C'$ could easily result in overcharging. 
Still, we will later use this lemma in our total charging argument for $Y$.

  \begin{lemma}\label{lem:single-cluster}
  For $C$ a cluster close to $u$ and $Y_u^C = \{v \in C \mid uv \in Y, v\text{ arrives after } u\}$, we have
       \[ \sum_{v \in Y_u^C} (1-d_{uv}) \leq 4 \cdot \max\{\nicefrac{1}{\rho}, \nicefrac{1}{(1-14\rho)}\}\cdot  \text{OPT}(u).\]
  \end{lemma}

  \begin{figure}
      \centering
      \includegraphics[width=0.7\linewidth]{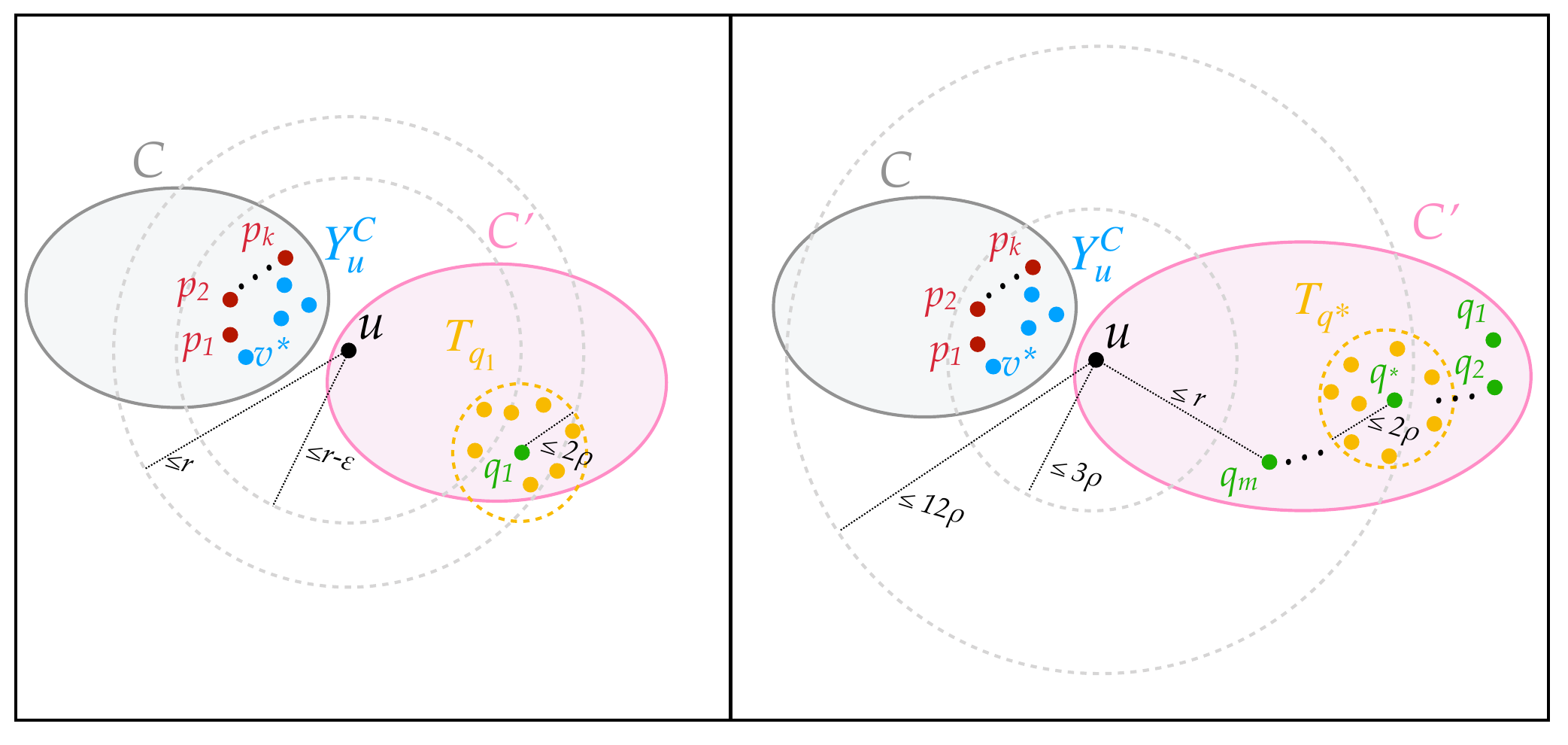}
      \caption{Charging arguments in Lemma \ref{lem:single-cluster} for Case 2. $C$ is close to $u$ because all pivots $p_i$ (red) of $C$ between the arrivals of $v^*$ and $u$ have $d_{up_i} \leq 3 \rho.$
      \textbf{(Left)} Here $v^*$ and $u$ arrive during the same phase of $C'$. The pivot of $C'$ is $q_1$, and we can charge $uv$ for $v \in Y_u^C$ (blue) to $uz$ for $z \in T_{q_1}$ (yellow). \textbf{(Right)} $u$ and $v^*$ arrive in different phases of $C'$. Now, $q_1$ may be too far away from $u$ but one can look among the pivots of $C'$ (green)  to find a pivot $q^*$ that's not too close and not too far from $u$. Then, charge $uv$ for $v \in Y_u^C$ (blue) to $uz$ for $z \in T_{q^*}$ (yellow).  This case relies on the fact that consecutive pivots move gracefully (Claim \ref{claim: pivot-migration}).}
      \label{fig:y-case}
  \end{figure}
  
\begin{proof}[Proof of Lemma \ref{lem:single-cluster}]
Let $v^*$ be the last node in $Y_u^C$ to arrive, and $C'$ be the cluster containing $u$.  Further, let $P_u^{C'}=\{q_1,\ldots,q_m\}$ be the pivots of $C'$ between when $v^*$ arrived and when $u$ arrived, and let $P_u^C=\{p_1,\ldots,p_k\}$ be those of $C$. As in the previous subsection, when $v^*$ initiates cluster $C$, we take $p_1=v^*$, and otherwise $p_1$ is the pivot that clusters $v^*$. Note also that $P_u^{C'} \neq \emptyset$, since $u$ cannot be starting cluster $C'$ when it arrives by definition of $uv \in Y$. 

We partition into two cases, based on whether $C$ or $C'$ had priority with respect to the ordering $\succ$ when $v^*$ arrived.
In both cases, we will show there is a pivot $q^*$ for $C'$ such that $|Y_u^C| \leq 4 \cdot |T_{q^*}|$, and for all $z \in T_{q^*}$, we have $d_{uz} \geq \rho$ and $1-d_{uz} \geq 1-14\rho$. 

\setcounter{case}{0}
 \begin{case}
  When $v^*$ arrived, either $C'$ did not exist, or $C'$ did exist and at that point in time $C \succ C'$.
  \end{case}

   Let $\ell$ be the phase for $C'$ when $C'$ clusters $u$, and note by assumption of the case that $\ell>0$.
  Node $p_k$ is the pivot of $C$ when $u$ arrives, and $d_{u p_k} \leq r$ since $C$ causes type $Y$ disagreements for $u$ by assumption.
  Therefore, it must be that $C' \succ C$ when $C'$ is in phase $\ell$. By
assumption of the case, $C\succ C'$ at some point before $C'$ enters phase $\ell$: If $C'$ did exist when $v^*$ arrived, then this follows from the assumption of the case and the fact that $v^*$ arrives before $u$. If $C'$ did not exist when $v^*$ arrived, on the other hand, then this implies $C$ was opened before $C'$, so for some period of time before $C'$ enters phase $\ell$, it holds that $C \succ C'$. In either scenario, we may thus define $0 \leq \ell'<\ell$ to be the last phase for cluster $C'$ where $C \succ C'$. Then, we apply Claim \ref{claim: far-swap} and see that for $p$ the pivot of $C$ in phase $\ell'$ for $C$ and $q$ the pivot of $C'$ in phase $\ell'+1$ for $C'$, $d_{pq} \geq 7 \cdot \rho.$ Since $d_{up} \leq 3 \cdot \rho$ (as $p \in P_u^C$ by assumption of the case and $C$ is a cluster close to $u$), we have that $d_{uq} \geq d_{pq} - d_{up} \geq 4 \cdot \rho$. 

If $\ell'+1 = \ell$, then we know $4 \rho \leq d_{uq} \leq r$, since then $q$ is precisely the pivot that clusters $u$ by definition of $\ell$. So for any $z \in T_q$ we have $d_{uz} \geq d_{uq} - d_{zq} \geq 4\rho - 2\rho = 2\rho$ and $1-d_{uz} \geq 1-d_{uq} - d_{qz} \geq 1-r-2\rho = 1-11\rho$. Moreover, $|T_q| \geq 2^{\ell'} \geq \nicefrac{1}{2}\cdot |Y_u^C| $, where the first inequality is by Claim \ref{claim: dense-pivot}, and the second inequality is due to the fact that $v^*$ is the last vertex to arrive in $Y_u^C$, and is clustered before or during phase $\ell'$ for $C$. Thus, one can choose $q^*=q$ to be the desired pivot.
 
Otherwise, $\ell'+1 < \ell$, and by Claim  \ref{claim: walking-trick} there is a pivot $q^*$ during phase $\phi(q^*) \in [\ell'+1,\ell)$ for cluster $C'$ such that for any $z \in T_{q^*}$, $d_{uz} \geq \rho$ and $1-d_{uz} \geq 1-14\rho$. Moreover, $|T_{q^*}| \geq 2^{\phi(q^*)-1} \geq 2^{\ell'} \geq \nicefrac{1}{2} \cdot |Y_u^C|$, again by Claim \ref{claim: dense-pivot} and the fact that $v^*$ is clustered before or during phase $\ell'$ for $C$.

\begin{case} When $v^*$ arrived, $C' \succ C$. 
\end{case}

  In this case,
  it must be that $v^*$ is outside of $q_1$'s clustering radius (recall $q_1$ is the pivot active for $C'$ when $v^*$ arrives), specifically $d_{v^*q_1} >r$. Therefore $d_{uq_1} \geq d_{v^*q_1}-d_{uv^*} \geq r-\varepsilon$.

  If $\phi(q_1) = \phi(q_m)$, then $ r-\varepsilon \leq d_{u{q_1}} \leq r$. Therefore for $z \in T_{q_1}^{C'}$, $1-d_{uz} \geq 1-d_{u{q_1}}-d_{uz} \geq 1-r-2 \rho = 1-11\rho$ and $d_{uz} \geq d_{u{q_1}} -d_{uz} \geq r-\varepsilon-2 \rho = 7\rho - \varepsilon$. 
  Further, we have
 $|Y_u^C| \leq 2 \cdot 2^{\phi(p_1)} \leq 2 \cdot 2^{\phi(q_1)} \leq 4 \cdot |T_{q_1}|$, where the second inequality follows from the fact that $C' \succ C$ when $v^*$ arrives, and the third inequality is by Claim \ref{claim: dense-pivot}. Here, we choose $q^*=q_1$ to be the desired pivot.

Otherwise, $\phi(q_1) < \phi(q_m)$. By Claim \ref{claim: walking-trick}, there is a pivot $q^*$ during phase $\phi(q^*) \in [\phi(q_1), \phi(q_m))$ for cluster $C'$ such that for any $z \in T_{q^*}$, $d_{uz} \geq \rho$ and $1-d_{uz} \geq 1-14\rho$. Moreover, $|Y_u^C| \leq 2 \cdot 2^{\phi(p_1)} \leq  2 \cdot 2^{\phi(q_1)} \leq 2 \cdot 2^{\phi(q^*)} \leq 4 \cdot |T_{q^*}|$,  where the first inequality follows from the fact that $C' \succ C$ when $v^*$ arrives, thus $\phi(p_1) \leq \phi(q_1) \leq \phi(q^*)$, and the second inequality is by Claim \ref{claim: dense-pivot}.

\medskip

Combining the two cases, we see that there is always a pivot $q^*$ for $C'$ such that $|Y_u^C| \leq 4 \cdot |T_{q^*}|$, and for all $z \in T_{q^*}$, we have $d_{uz} \geq \rho$ and $1-d_{uz} \geq 1-14\rho$. Thus,
\begin{align*}
    \sum_{v \in Y_u^C} (1-d_{uv})& \leq |Y_u^C| \\
    &\leq 4 \cdot |T_{q^*}|  \\
    &\leq 4 \cdot \Bigg (\sum_{\substack{z \in T_{q^*}:\\ \mathcal{C}^*(u) = \mathcal{C}^*(z) }}\nicefrac{1}{\rho}\cdot d_{uz}+  \sum_{\substack{z \in T_{q^*}:\\ \mathcal{C}^*(u) \neq \mathcal{C}^*(z) }}\nicefrac{1}{(1-14\rho)}\cdot (1-d_{uz}) \Bigg )\\
    & \leq 4 \cdot \max\{\nicefrac{1}{\rho}, \nicefrac{1}{(1-14\rho)}\}\cdot  \text{OPT}(u).
\end{align*}

  \end{proof}

When there are $t > 1$ clusters close to $u$, for $u \in C'$, we cannot charge all the $uv \in Y$ to $uz$ for $z \in C'$, as this would result in overcharging edges in $C'$. Suppose $Y_u^C = \{v \in C \mid uv \in Y, v \text{ arrives after }u \}$ is non-empty for multiple clusters $C$ that are close to $u$. Label these clusters in order of the last node to arrive in $Y_u^C$, 
i.e., we label them $C^1,\ldots,C^t$, where for $v_i^*$ the last node to arrive in $Y_u^{C^i}$ for $i \in [t]$, $v_1^*$ arrives before $v_2^*$, which arrives before $v_3^*$, etc.
For $i \in [t-1]$, we will charge the disagreements incurred between $u$ and $Y_u^{C^i}$ to disagreements incurred between $u$ and the dense ball of some pivot in $C^{i+1}$ (i.e., we show there is a pivot $q^*_i$ of $C^{i+1}$ such that $|T_{q^*_i}| = \Omega(|Y_u^{C^i}|)$ and $d_{uz}, 1-d_{uz} = \Omega(1)$ for each $z \in T_{q^*_i}$), and then charge the disagreements between $u$ and $C^t$ to disagreements between $u$ and $C'$ by invoking Lemma \ref{lem:single-cluster}. 
 See Figure \ref{fig:y-k}.

\begin{figure}
    \centering
    \includegraphics[width=0.4\linewidth]{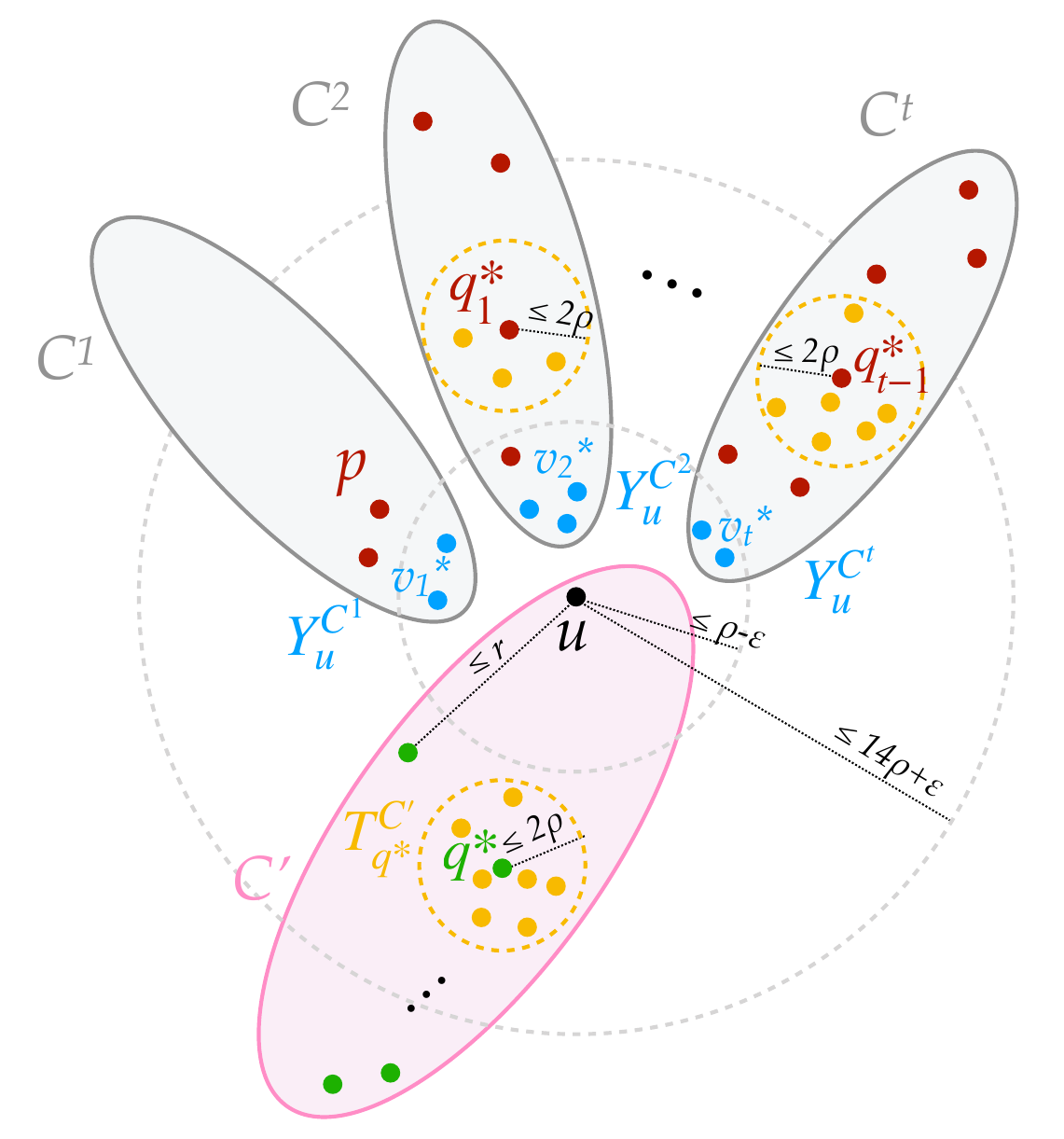}
    \caption{When many clusters $C^i$ are close to $u$ and have non-empty $Y_u^{C^i}$, the charging is more complicated. Let $C^1,\ldots,C^t$ be clusters close to $u$ with nonempty $Y_u^{C^i}$. We separate $u$ from all nodes in $ \cup_{i \in [t]} Y_u^{C^i}$ (blue). For $i \in [t-1]$ and $v \in Y_u^{C^i}$, we charge the cost of $uv$ to $uz$ for $z \in T_{q^*_{i}}$ (yellow), for specially chosen pivot $q^*_{i} \in C^{i+1}$, and charge $uv$ for $v \in Y_u^{C^t}$ to $uz$ for $z \in T_{q^*}$, for specially chosen pivot $q^* \in C'$.}
    \label{fig:y-k}
\end{figure}
\begin{lemma} \label{lem: k-clusters}
   \[\sum_{\substack{uv \in Y:\\ u \leadsto \mathcal{C}(v) }}(1-d_{uv})\leq  16 \cdot \max\{\nicefrac{1}{(\rho-\varepsilon)}, \nicefrac{1}{(1-14\rho-\varepsilon)}\} \cdot \text{OPT}.\]
\end{lemma}
\begin{proof}[Proof of Lemma \ref{lem: k-clusters}]

Fix $u$ in cluster $C'$ and suppose there are $t$ clusters close to $u$. Label these clusters in order of the last node to arrive in $Y_u^C$, 
i.e., we label them $C^1,\ldots,C^t$, where for $v_i^*$ the last node to arrive in $Y_u^{C^i}$ for $i \in [t]$, $v_1^*$ arrives before $v_2^*$, which arrives before $v_3^*$, etc. 

Fix $i \in [t-1]$. We case on whether or not at the time when $v^*$ arrives, $C^{i+1} \succ C^i$.

  \setcounter{case}{0}
\begin{case} When $v^*_i$ arrived, either $C^{i+1}$ did not exist, or $C^{i+1}$ did exist and at that point in time $C^i \succ C^{i+1}$.
\end{case}

  Let $\hat{p}$ be the pivot of $C^i$ when $v^*_{i+1}$ arrived (by assumption of the case, $\hat{p}$ exists because $v_i^*$ is in $C^i$ and arrives before $v_{i+1}^*$). Note that $v^*_{i+1}$ is within distance $r$ of $\hat{p}$ because $\hat{p}$ is a close pivot to $u$, so
  $d_{v^*_{i+1}\hat{p}} \leq d_{u\hat{p}} +d_{v^*_{i+1}u}\leq 3 \rho+\varepsilon \leq r.$ 

It is not possible for  $v^*_{i+1}$ to initiate cluster $C^{i+1}$, because if this were the case, then when $v_{i+1}^*$ arrives, $C^i \succ C^{i+1}$ by definition of the ordering and $d_{v_{i+1}^*\hat{p}} \leq r$, so $v_{i+1}^*$ would be clustered by $\hat{p}$ in $C^{i}$.
Therefore, $v_{i+1}^*$ is clustered into $C^{i+1}$ during some phase, call it $\ell$, for $C^{i+1}$.
 
 It must be that $C^{i+1} \succ C^i$ when $C^{i+1}$ enters phase $\ell$ since $\hat{p}$ is close enough to $v_{i+1}^*$ to cluster it.  (Note that $C^i$ is still active when $C^{i+1}$ enters phase $\ell$, for if not, then $C^i$ is no longer active when $v_{i+1}^*$ is clustered, thus no longer active when $u$ is clustered, contradicting that $uv_{i+1}^* \in Y$.) By assumption of the case, $C^i \succ C^{i+1}$ at some point before $C^{i+1}$ enters phase $\ell$: If $C^{i+1}$ did exist when $v_i^*$ arrived, then this follows from the assumption of the case and the fact that $v_i^*$ arrives before $v_{i+1}^*$. If $C^{i+1}$ did not exist when $v_i^*$ arrived, on the other hand, then this implies $C^i$ was opened before $C^{i+1}$, so for some period of time before $C^{i+1}$ enters phase $\ell$, it holds that $C_i \succ C_{i+1}$.  In either scenario, we see that $\ell>0$ and may thus define $\ell' < \ell$ to be the last phase for cluster $C^{i+1}$ where $C^i \succ C^{i+1}$. Let $p$ be the pivot of $C^{i}$ in phase $\ell'$ for $C^i$, and let $q$ be the pivot of $C^{i+1}$ in phase $\ell'+1$ for $C^{i+1}$. By Claim \ref{claim: far-swap}, $d_{pq} \geq 7\rho$. As $C^i$ is a cluster close to $u$, $d_{up} \leq 3\rho$. In turn, $d_{uq} \geq d_{pq} - d_{up} \geq 4\rho$, and $d_{v_{i+1}^*q} \geq d_{uq} - d_{v_{i+1}^*u} \geq 4\rho - \varepsilon$.

If $\ell'+1=\ell$, then we know $4\rho \leq d_{uq} \leq r+\varepsilon$, where for the upper bound we have used that $q$ is precisely the pivot that clusters $v_{i+1}^*$ by definition of $\ell$, so $d_{uq} \leq d_{v_{i+1}^*u} + d_{v_{i+1}^*q} \leq \varepsilon + r$. So for any $z \in T_q$ we have $d_{uz} \geq d_{uq} - d_{zq} \geq 4\rho - 2\rho = 2\rho$ and $1-d_{uz} \geq 1-d_{uq} - d_{qz} \geq 1-r-\varepsilon-2\rho$. Moreover, $|T_q| \geq 2^{\ell'} \geq \nicefrac{1}{2}\cdot |Y_u^{C^i}| $, where the first inequality is by Claim \ref{claim: dense-pivot}, and the second inequality is due to the fact that $v_i^*$ is the last vertex in $Y_u^{C^i}$, and is clustered during or before phase $\ell'$ for $C^i$.

Otherwise, $\ell'+1 < \ell$, and by Claim  \ref{claim: walking-trick} there is a pivot $q^*$ during phase $\ell^* \in [\ell'+1,\ell)$ for cluster $C^{i+1}$ such that for any $z \in T_{q^*}$, $d_{v_{i+1}^*z} \geq \rho$ and $1-d_{v_{i+1}^*z} \geq 1-14\rho$. So for any $z \in T_{q^*}$, we have $d_{uz} \geq \rho -\varepsilon$ and $1-d_{uz} \geq 1-14\rho-\varepsilon$. Moreover, $|T_{q^*}| \geq 2^{\ell^*-1} \geq 2^{\ell'} \geq \nicefrac{1}{2} \cdot |Y_u^{C^i}|$, again by Claim \ref{claim: dense-pivot} and the fact that $v_i^*$ is clustered during or before phase $\ell'$ for $C^i$.

    \begin{case} When $v^*_{i}$ arrived, $C^{i+1} \succ C^i$. 
    \end{case}
    
    Let $q$ be the pivot of $C^{i+1}$ when $v^*_{i}$ arrived, which must exist by the case hypothesis.
  In this case, it must be that $v^*_i$ is outside of $q$'s clustering radius, specifically $d_{v^*_{i}q} >r$. Therefore $d_{uq} \geq d_{v^*_{i}q}-d_{uv^*_{i}} \geq r-\varepsilon$, and $d_{v_{i+1}^* q}\geq d_{uq} - d_{uv_{i+1}^*} \geq r-2\varepsilon$. 

  By the case hypothesis and the fact that $v_i^*$ arrives before $v_{i+1}^*$, it follows that $v_{i+1}^*$ cannot initiate cluster $C^{i+1}$.
  If $v_{i+1}^*$ is clustered during phase $\phi(q)$ for $C^{i+1}$, then we have $d_{uq} \geq r-\varepsilon$ by the above and $1- d_{uq} \geq 1- r - \varepsilon$. So for any $z \in T_q$, we have $d_{uz} \geq d_{uq} - d_{qz} \geq r-\varepsilon - 2\rho = 7\rho-\varepsilon$ and $1-d_{uz} \geq 1-d_{uq} - d_{qz} \geq 1-r-\varepsilon - 2\rho = 1-11\rho-\varepsilon$. Moreover, $|T_q| \geq 2^{\phi(q)-1} \geq \nicefrac{1}{4} \cdot |Y_u^{C^i}|$, where the first inequality is by Claim \ref{claim: dense-pivot}, and the second inequality is due to $C^{i+1} \succ C^i$ implying that $2 \cdot 2^{\phi(q)} = 2 \cdot |C^{i}_{\phi(q)}| \geq |Y_{u}^{C^i}|$.

Otherwise, we apply Claim \ref{claim: walking-trick} to find a pivot $q^*$ for $C^{i+1}$ with $\phi(q^*) \geq \phi(q)$ such that for any $z \in T_{q^*}$, we have $d_{v_{i+1}^*z} \geq \rho$ and $1-d_{v_{i+1}^*z} \geq 1-14\rho$. So for any $z \in T_{q^*}$, $d_{uz} \geq \rho - \varepsilon$ and $1-d_{v_{i+1}^*z} \geq 1-14\rho-\varepsilon$. Moreover, $|T_{q^*}| \geq 2^{\phi(q^*)-1} \geq \nicefrac{1}{4} \cdot |Y_u^{C^i}|$, where the first inequality is by Claim \ref{claim: dense-pivot}, and the second inequality is due to $C^{i+1} \succ C^i$ implying that $2 \cdot 2^{\phi(q^*)} \geq 2 \cdot 2^{\phi(q)} = 2 \cdot |C^{i}_{\phi(q)}| \geq |Y_{u}^{C^i}|$.  

Combining the cases above, we conclude that for each $i \in [t-1]$, there is a pivot $q_i^*$ for $C^{i+1}$ such that for all $z \in T_{q_i^*}$, we have that $d_{uz} \geq \rho - \varepsilon$ and $1-d_{uz} \geq 1-14\rho-\varepsilon$, and further that $|Y_u^{C^i}| \leq 4 \cdot |T_{q^*}|$. 

All together, we have 
  \begin{align*}
      \sum_{i=1}^{t} \sum_{v \in Y_u^{C^i}}(1-d_{uv}) &\leq \sum_{i=1}^{t-1} \sum_{v \in Y_u^{C^i}}(1-d_{uv})+\sum_{v \in Y_u^{C^t}}(1-d_{uv})\\
      & \leq \sum_{i=1}^{t-1}|Y_u^{C^i}|+\sum_{v \in Y_u^{C^t}}(1-d_{uv})\\
      & \leq 4 \cdot \sum_{i=1}^{t-1}|T_{q_i^*}|+\sum_{v \in Y_u^{C^t}}(1-d_{uv})\\
      &\leq 4 \cdot \sum_{i=1}^{t-1}\sum_{z \in T_{q_i^*}}1+\sum_{v \in Y_u^{C^t}}(1-d_{uv})\\
 &\leq 4 \cdot \sum_{i=1}^{t-1} \bigg (\sum_{\substack{z \in T_{q_i^*}:\\ \mathcal{C}(u) = \mathcal{C}(z)}}\nicefrac{1}{(\rho-\varepsilon)} \cdot d_{uz} +\sum_{\substack{z \in T_{q_i^*}:\\ \mathcal{C}(u) \neq \mathcal{C}(z)}}\nicefrac{1}{(1-14\rho-\varepsilon)} \cdot (1-d_{uz}) \Bigg )+\sum_{v \in Y_u^{C^t}}(1-d_{uv})\\
      &\leq 4 \cdot \max\{\nicefrac{1}{(\rho-\varepsilon)}, \nicefrac{1}{(1-14\rho-\varepsilon)}\}\cdot  \text{OPT}(u) + 4 \cdot \max\{\nicefrac{1}{\rho}, \nicefrac{1}{(1-14\rho)}\}\cdot  \text{OPT}(u)
  \end{align*}
  where in the last line we have applied Lemma \ref{lem:single-cluster}. We note that $T_{q_i^*}$ is contained in $C^{i+1}$, so no $uz$ appears in the inner sums for multiple $i$ and thus the final inequality follows.
  
  To complete the proof, we sum over all $u$:
    \[  \sum_{\substack{uv \in Y:\\ u \leadsto\mathcal{C}(v) }} (1-d_{uv}) = \sum_u \sum_{\substack{C \in \mathcal{C}:\\ u \leadsto C}} \sum_{v \in Y_u^C} (1-d_{uv}) \leq  16 \cdot \max\{\nicefrac{1}{(\rho-\varepsilon)}, \nicefrac{1}{(1-14\rho-\varepsilon)}\}\cdot \text{OPT}. \]
\end{proof}

Combining Lemmas \ref{lem: caseY-C-far-from-u} and \ref{lem: k-clusters}, we are now able to prove Lemma \ref{lem:y-charge}.

\begin{proof}[Proof of Lemma \ref{lem:y-charge}]
    By Lemmas \ref{lem: caseY-C-far-from-u} and \ref{lem: k-clusters}, 

\begin{align*}
 \sum_{uv \in Y}(1-d_{uv}) &\leq \sum_{\substack{uv \in Y:\\ u \not\leadsto \mathcal{C}(v) }}(1-d_{uv}) +  \sum_{\substack{uv \in Y:\\ u \leadsto\mathcal{C}(v) }} (1-d_{uv}) \\
 &\leq 8 \cdot \max\{\nicefrac{1}{\rho}, \nicefrac{1}{(1-11\rho - \varepsilon)}\} \cdot \text{OPT} + 16 \cdot \max\{\nicefrac{1}{(\rho-\varepsilon)}, \nicefrac{1}{(1-14\rho-\varepsilon)}\} \cdot \text{OPT}
 \end{align*}
\end{proof}

\subsubsection{Charging $uv \in Z$}\label{sec:Z}

For $uv \in Z$ where $u$ arrives after $v$, $u$ is not clustered with $v$ in cluster $C$ because $C$ was inactive when $u$ arrives.  
We will show that the fact that $C$ is inactive certifies that \emph{no node of $C$ is dense}. 
This forces a constant fraction of $C$ to be far from $u$. 

One difference in bounding the cost of type $Z$ edges is that in some cases, we rely on a set of dense balls around pivots, instead of just one.

\begin{lemma}{\label{lem:z-charge}}
    \[\sum_{uv \in Z}(1-d_{uv}) \leq  \nicefrac{8}{\rho} \cdot \max \left \{\nicefrac{4}{\rho}, \nicefrac{1}{(1-\varepsilon - 57\rho)} \right\} \cdot  \text{OPT}. \]
\end{lemma}
\begin{proof}[Proof of Lemma \ref{lem:z-charge}]
Fix a vertex $u$, and a cluster $C$ that is inactive when $u$ arrives. Let $Z^C_u = \{v \in C \mid uv \in Z, v\text{ arrives before }u\}$. Let $v^*$ be the last node to arrive in $ Z^C_u$.

Define $\widetilde{P}$ to be the ordered multiset of pivots for $C$ \emph{strictly after} $v^*$ is clustered.\footnote{Note this is slightly different than the set of pivots, $P_u^C$, we tracked in the analysis of charging the $X$ and $Y$ edges. While the set $P_u^C$ was never empty, $\widetilde{P}$ can be empty.} We note some special cases. If $v^*$ starts the cluster $C$, then the first pivot in $\widetilde{P}$, with respect to the ordering of $\widetilde{P}$, is $v^*$, since $v^*$ is the pivot for phase 0  which begins after $v^*$ is clustered. We also note that, in the case that $v^*$ does \emph{not} start the cluster, it is possible that $\widetilde{P}$ is actually empty. 

When $\widetilde{P} \neq \emptyset$, define $p_{\text{first}}$ to be the first pivot in $\widetilde{P}$ and $p_{\text{last}}$ to be the last pivot in $\widetilde{P}$ with respect to the ordering of $\widetilde{P}$ (with $p_{\text{first}} = p_{\text{last}}$ in the case that $|\widetilde{P}|=1$). Then $\phi(p_{\text{last}})$ is the last phase for $C$, and the final cluster $C$ consists of \emph{exactly} $2^{\phi(p_{\text{last}})+1}$ nodes, the last $2^{\phi(p_{\text{last}})}$ of which are clustered into $C$ by $p_{\text{last}}$ before $C$ becomes inactive.

We partition into three cases, based on whether there is a pivot in $\widetilde{P}$ that is sufficiently far from $u$ (distance at least $2r$), and whether $v^*$ is dense enough to be a pivot in some phase $\phi(p)$ for $p \in \widetilde{P}$.
In each of the three cases, (as in charging the $X$ and $Y$ edges) we will show that there exists a set $S_u^C$ such that $|S_u^C| = \Omega(|Z_u^C|)$ and for each $z \in S_u^C$, both $d_{uz} = \Omega(1)$ and $1-d_{uz} = \Omega(1)$. This will in turn allow us to charge the cost of all $uv$ with $v \in Z_u^C$ to $\text{OPT}(u)$.

\setcounter{case}{0}

\begin{case}
$\widetilde{P} \neq \emptyset$ and there is some $p \in \widetilde{P}$ with $d_{up} > 2r$.
\end{case}

Let $p^*$ be chosen maximally (with respect to the ordering on $\widetilde{P}$) such that the pivots from $p_{\text{first}}$ through $p^*$ all are in $u$'s $2r$-ball. 
 
\begin{claim} 
As defined above, $p^*$ must exist.
\end{claim}

\begin{proof}
We need to show that some $p \in \widetilde{P}$ is in $u$'s $2r$-ball. We will show that $p_{\text{first}}$ specifically is in $u$'s $2r$-ball. To see why this is true, we consider two cases. 

First consider the case that $v^*$ starts $C$. Then $p_{\text{first}} = v^*$ and of course $d_{uv^*} \leq \varepsilon \leq 2r$, and we are done.

Otherwise, $v^*$ was clustered by some pivot, call it $p'$, and by construction $p' \not \in \widetilde{P}$. Since $\widetilde{P} \neq \emptyset$ by hypothesis of the case, $p_{\text{first}}$ is the pivot immediately following $p'$ in the ordering of $C$'s pivots, that is, $\phi(p_{\text{first}}) = \phi(p')+1$. So $d_{up_{\text{first}}} \leq d_{uv^*} + d_{v^*p_{\text{first}}} \leq \varepsilon + d_{v^*p'} + d_{p'p_{\text{first}}} \leq \varepsilon + r + 3\rho$, where in the last inequality we have used Claim \ref{claim: pivot-migration}. Since $\varepsilon + r + 3 \rho \leq 2r$, the claim holds for this case, thus holds overall.
\end{proof}
 
 Moreover, by assumption of the case, $p^* \neq p_{\text{last}}$, so there is a pivot immediately after $p^*$ in $\widetilde{P}$, and this pivot is outside of $u$'s $2r$-ball as $p^*$ is chosen maximally; call this pivot $p^*_+$.

 We now show that taking $S_u^C = T_{p^*}$ meets the requirements for $S_u^C$ above. To see this, we observe that $d_{up^*} \geq d_{up^*_+} - d_{p^*p^*_+} \geq 2r - 3\rho$ by Claim \ref{claim: pivot-migration}. By our choice of $p^*$, we also have $d_{up^*} \leq 2r$. So for all $z \in T_{p^*}$, we have $d_{uz} \geq d_{up^*} - d_{p^*z} \geq 2r - 3 \rho -2\rho = 13\rho$, and $d_{uz} \leq d_{up^*} + d_{p^*z} \leq 2r+2\rho = 20\rho$. To see that $|T_{p^*}|$ is sufficiently large compared to $|Z_u^C|$, we observe that $Z_u^C \subseteq C_{\phi(p^*)}$, so $|Z_u^C| \leq 2^{\phi(p^*)} \leq 2 \cdot |T_{p^*}|$, where the last inequality is by Claim \ref{claim: dense-pivot}. Thus we may take $S_u^C = T_{p^*}$ in this case.

\medskip

\begin{case}
    $\widetilde{P} = \emptyset$, or $\widetilde{P} \neq \emptyset$ but $v^*$ is \emph{not} dense enough to be a pivot in any of the phases  $\phi(p_{\text{first}}),\dots, \phi(p_{\text{last}})$ (that is, for each $\ell \in \{\phi(p_{\text{first}}), \dots, \phi(p_{\text{last}})\}$, it holds that $\sum_{z \in C_\ell} d_{zv^*} > \rho \cdot |C_\ell|$). 
\end{case}

Let $C$ here refer to the \emph{final} cluster, that is, the cluster $C$ at the time of its deactivation. 
We define $C_{\text{far}}$ to be $C$ when $\widetilde{P}$ is empty, and to be $C_{\phi(p_{\text{first}})}$ otherwise. In either case,
\begin{align}\sum_{z \in C_{\text{far}} }d_{zv^*} > \rho \cdot |C_{\text{far}}|.
\label{eqn:not-dense}\end{align} 
We will show that inequality (\ref{eqn:not-dense}) implies that many nodes in $C_{\text{far}}$ are far away from $v^*.$ In particular, it must be that 
\begin{equation} \label{eqn:Cfar-close-small}
|C_{\text{far}} \cap \text{Ball}(v^*,\nicefrac{\rho}{2})| < (1-\nicefrac{\rho}{2})\cdot |C_{\text{far}}|,
\end{equation}
as otherwise 
\[\sum_{z \in C_{\text{far}}}d_{zv^*} \leq \sum_{\substack{z \in C_{\text{far}}:\\ z \in \text{Ball}(v^*, \nicefrac{\rho}{2})}} \nicefrac{\rho}{2} + \sum_{\substack{z \in C_{\text{far}}:\\ z \not \in \text{Ball}(v^*, \nicefrac{\rho}{2})}} 1 \leq \nicefrac{\rho}{2} \cdot |C_{\text{far}}| + 1 \cdot \nicefrac{\rho}{2} \cdot |C_{\text{far}}| = \rho \cdot |C_{\text{far}}| \]
which would contradict inequality (\ref{eqn:not-dense}).

We define \[O_u^C := \{z \in C_{\text{far}} \mid z \not \in \text{Ball}(v^*, \nicefrac{\rho}{2}) \}.\] 
We established in line (\ref{eqn:Cfar-close-small}) that $|O_u^C| \geq \nicefrac{\rho}{2} \cdot |C_{\text{far}}|.$  

Let $\alpha \in \mathbb{Z}^+$ be minimal such that $2^{-\alpha} \leq \rho$. 
Note that $|C_{\text{far}}|$ is always a power of 2: when
$\widetilde{P} \neq \emptyset$, we have $|C_{\text{far}}| =
|C_{\phi(p_{\text{first}})}| = 2^{\phi(p_{\text{first}})}$ by definition,
and when $\widetilde{P} = \emptyset$, the cluster $C_{\text{far}} = C$ was
deactivated at a doubling check, so its final size is a power of 2.
Exactly one of the following two subcases holds.

\medskip

\begin{figure}
    \centering
    \includegraphics[width=0.8\linewidth]{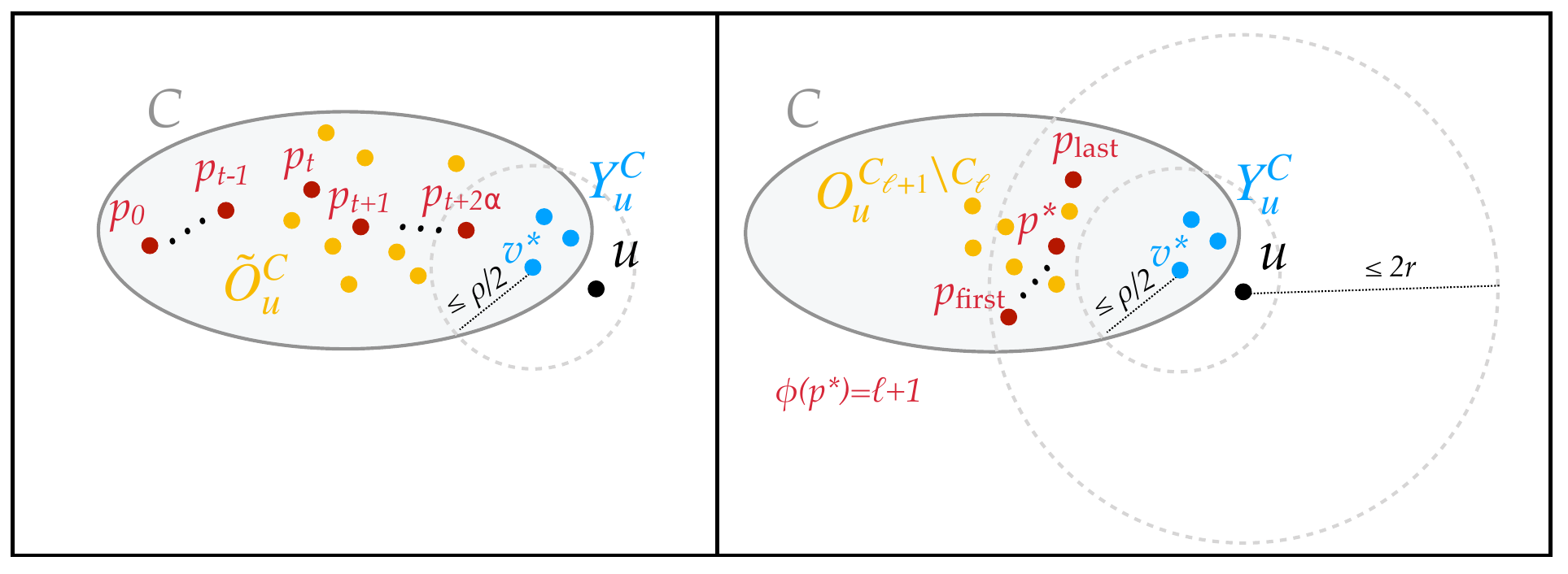}
    \caption{\textbf{(Left)} Illustrates Subcase 2.1 of Case 2. In this case,  $p_{\text{first}} = p_{t+2\alpha}$ is the first pivot of $C$ after $v^*$ is clustered, $v^*$ is not dense enough to be a pivot, and  $|C_{\phi(p_{\text{first}})}| \geq 2^{2 \alpha}$ (for $\alpha \in \mathbb{Z}^+$ minimal with $2^{-\alpha} \leq \rho$).
    There are many nodes in $C_{\phi(p_{\text{first}})}$ far from $v^*$, and it suffices to consider $\widetilde{O}_u^C$ (yellow), the nodes in $C_{\phi(p_{\text{first}})}$ clustered in the most recent $2 \alpha$ phases. \textbf{(Right)} Illustrates Case 3. In this case, all pivots of $C$ after $v^*$ arrives are within distance $2r$ of $u$ and $v^*$ is dense enough to be a pivot in phase $\ell = \phi(p^*)$, but not in phase $\ell+1$. Since $v^*$ went from being sufficiently dense to not being sufficiently dense to be a pivot for $C$ between phases $\ell$ and $\ell+1$, it must be that many nodes in $C_{\ell+1} \setminus C_\ell$ are far from $v^*$ (and thus far from $u$).
    }
\label{fig:charge-pos-z}
\end{figure}

\begin{subcase} \label{subcase:many-pivots} $|C_{\text{far}}| \geq 2^{2\alpha}$.
\end{subcase}

We show in this subcase that there is a large subset $\widetilde{O}_u^C \subseteq O_u^C$ such that taking $S_u^C = \widetilde{O}_u^C$ meets the requirements for $S_u^C$ above.  

The subcase implies that there are at least $2\alpha$ phases of $C$ before $C$ deactivates: phases $0, 1, \dots, 2\alpha-1$. Define 
\[\widetilde{O}_u^C := O_u^C \setminus C_{\log|C_{\text{far}}|-2\alpha} \]
(note that $\log |C_{\text{far}}|$ is integral as observed above).
As desired, $\widetilde{O}_u^C$ is sufficiently large compared to $|Z_u^C|$:
\begin{align*}
|\widetilde{O}_u^C| \geq |O_u^C| - |C_{\log|C_{\text{far}}|-2\alpha}| = |O_u^C| - 2^{\log|C_{\text{far}}|-2\alpha} &= |O_u^C| - |C_{\text{far}}| \cdot \nicefrac{1}{2^{2\alpha}} \\
&\geq \nicefrac{\rho}{2} \cdot |C_{\text{far}}| - |C_{\text{far}}| \cdot \nicefrac{1}{2^{2\alpha}} \\
&\geq \nicefrac{\rho}{2} \cdot |C_{\text{far}}| - \rho^2 \cdot |C_{\text{far}}| \\
&\geq \nicefrac{\rho}{4} \cdot |C_{\text{far}}| \\
&\geq \nicefrac{\rho}{4} \cdot |Z_u^C|
\end{align*}

The penultimate line follows from the fact that $\rho \leq \nicefrac{1}{4}$. To see that the last line holds, we consider two cases: (1) if $C_{\text{far}} = C$, then the last line holds automatically, and (2) otherwise, $C_{\text{far}} = C_{\phi(p_\text{first})}$ which includes $v^*$ and all nodes arriving before $v^*$ by definition, thus includes all of  $Z_u^C$, so the last line again holds. 

Now we lower bound $d_{uz}$ and $1-d_{uz}$ for $z \in \widetilde{O}_u^C$. Since $\widetilde{O}_u^C \subseteq O_u^C$, we have $d_{uz} \geq d_{v^*z}-d_{uv^*} \geq \nicefrac{\rho}{2} - \varepsilon$ for any $z \in \widetilde{O}_u^C$. 

To obtain a lower bound on $1-d_{uz}$, equivalently an upper bound on $d_{uz}$, for $z \in \widetilde{O}_u^C$, we use that $u$ is close to $v^*$, and in turn that $v^*$ is close to $z$ owing to the fact that we pruned off a prefix of cluster $C$ from $O_u^C$ to define $\widetilde{O}_u^C$. More formally, if $\widetilde{P} \neq \emptyset$, then $p_{\text{first}}$ is defined and $d_{v^* p_{\text{first}}} \leq r + 3\rho$. This is because either $v^*$ starts cluster $C$, in which case $p_{\text{first}} = v^*$ and thus $d_{v^* p_{\text{first}}} = 0$; or, $v^*$ is clustered by a pivot, call it $p'$, thus $d_{v^*p_{\text{first}}} \leq d_{v^* p'} + d_{p' p_{\text{first}}} \leq r + 3\rho$ by Claim \ref{claim: pivot-migration}. 

In turn, since $z \not \in C_{\log |C_{\phi(p_\text{first})}|-2\alpha}$, $z$ has a pivot $p(z)$ in $C$ and there are at most $2\alpha$ pivots from $p(z)$ to $p_{\text{first}}$, not counting the former. So by Claim \ref{claim: pivot-migration}, $d_{p_{\text{first}}p(z)} \leq 2\alpha \cdot 3\rho$. For $\widetilde{P} \neq \emptyset$, we have in total 
\begin{equation} \label{eqn:walk-pruned-cluster-nonempty}
    d_{uz} \leq d_{uv^*} + d_{v^*p_{\text{first}}} + d_{p_{\text{first}} p(z)} + d_{p(z) z} \leq \varepsilon + 2r + 3\rho \cdot (2\alpha + 1).
\end{equation}

On the other hand, if $\widetilde{P} = \emptyset$, then $v^*$ did \emph{not} start its own cluster, so $v^*$ is clustered by a pivot, call it $p'$. Since $z \not \in C_{\log |C|-2\alpha}$, $z$ has a pivot $p(z)$ in $C$ and there are at most $2\alpha-1$ pivots from  $p(z)$ to $p'$, not counting the former. So by Claim \ref{claim: pivot-migration}, $d_{p'p(z)} \leq (2\alpha-1) \cdot 3\rho$. So for $\widetilde{P} = \emptyset$, we have in total 

\begin{equation}\label{eqn:walk-pruned-cluster-empty}
    d_{uz} \leq d_{uv^*} + d_{v^*p'} + d_{p' p(z)} + d_{p(z) z} \leq \varepsilon + 2r + 3\rho \cdot (2\alpha-1).
\end{equation}

We have shown that $d_{uz}, 1-d_{uz} = \Omega(1)$ for $z \in \widetilde{O}_u^C$. Thus we may take $S_u^C =\widetilde{O}_u^C$ in this subcase. 

\medskip

\begin{subcase} \label{subcase:few-pivots} $|C_{\text{far}}| \leq  2^{2\alpha-1}$.
\end{subcase}

In this case, we can actually show that taking $S_u^C = O_u^C$ meets the requirements for $S_u^C$ above; we do not need to prune this set as we did in Subcase \ref{subcase:many-pivots}. The lower bound $d_{uz} \geq \nicefrac{\rho}{2} -\varepsilon$ holds just as in Subcase \ref{subcase:many-pivots}. For the lower bound on $1-d_{uz}$, equivalently the upper bound on $d_{uz}$, the same exact chains of inequalities (\ref{eqn:walk-pruned-cluster-nonempty}) and (\ref{eqn:walk-pruned-cluster-empty}) hold as in Subcase \ref{subcase:many-pivots}, by noting that since $|C_{\text{far}}| \leq 2^{2\alpha-1}$, every pivot of $C$ has phase at most $2\alpha-1$, so there are at most $2\alpha$ pivots from $p(z)$ to $p_{\text{first}}$ in the case that $p_{\text{first}}$ exists, and at most $2\alpha-1$ pivots between $p'$ and $p(z)$ otherwise. (For the node $z$ that initiated $C$, take $p(z) = z$, its phase-$0$ pivot, so $d_{p(z)z} = 0$; this only strengthens the bounds.) Thus $d_{uz}, 1-d_{uz} = \Omega(1)$ for $z \in O_u^C$.

Finally, we have from above that $|O_u^C| \geq \nicefrac{\rho}{2} \cdot |C_{\text{far}}|$ and $|C_{\text{far}}| \geq |Z_u^C|$, so $|O_u^C|$ is sufficiently large. Thus we may take $S_u^C = O_u^C$ in this subcase.

\medskip

\begin{case}
$\widetilde{P} \neq \emptyset$, all pivots $p \in \widetilde{P}$ have $d_{up} \leq 2r$, and $v^*$ is dense enough to be a pivot in at least one of the phases $\phi(p_{\text{first}}), \dots, \phi(p_{\text{last}})$ (that is, $\sum_{w \in C_\ell} d_{wv^*} \leq \rho \cdot |C_\ell|$ for some $\ell \in \{\phi(p_{\text{first}}), \dots, \phi(p_{\text{last}})\}$).
\end{case}

Let $\ell$ be the \emph{last} phase among $\phi(p_{\text{first}}), \dots, \phi(p_{\text{last}})$ such that $v^*$ is dense enough to be a pivot during phase $\ell$. We slightly abuse notation to cover an edge case: if $\ell=\phi(p_{\text{last}})$, we let $C_{\ell+1}$ denote the final cluster $C$, that is, the cluster at the end of phase $\phi(p_{\text{last}})$, which is when the algorithm checked whether there was a node dense enough to serve as a pivot for a potential phase $\phi(p_\text{last})+1$ for $C$, found none, and made $C$ inactive as a result.
The following inequalities always hold:
\begin{align*}
    \sum_{z \in C_{\ell}} d_{zv^*} \leq \rho \cdot 2^{\ell}, \qquad 
    \sum_{z \in C_{\ell+1}} d_{zv^*} > \rho \cdot 2^{\ell+1},
\end{align*}
where the second inequality follows immediately from the choice of $\ell$ in the case that $\ell < \phi(p_{\text{last}})$, and from the fact that cluster $C$ became inactive after phase $\ell$ in the case that $\ell = \phi(p_{\text{last}})$ (so $v^*$ in particular was not dense enough at the end of phase $\phi(p_{\text{last}})$ to serve as a pivot for a next phase).
Combining these inequalities, 
\begin{align}
    \sum_{z \in C_{\ell+1} \setminus C_{\ell}} d_{zv^*}=\sum_{z \in C_{\ell+1}}  d_{z v^*} -\sum_{z \in C_{\ell}} d_{zv^*}> \rho \cdot \left (2^{\ell+1}-2^{\ell} \right ) = \rho \cdot 2^{\ell}.\label{eqn:large-dist-pk}
\end{align}
In words, the average distance of the $2^\ell$ nodes in $C_{\ell+1} \setminus C_{\ell}$, i.e., the nodes that arrive in phase $\ell$, to $v^*$ is more than $\rho$. 

We claim the following upper bound on the number of nodes arriving during phase $\ell$ that are within distance $\nicefrac{\rho}{2}$ of $v^*$:
\begin{align}
    |(C_{\ell+1} \setminus C_{\ell}) \cap \text{Ball}(v^*, \nicefrac{\rho}{2} ) |< (1-\nicefrac{\rho}{2}) \cdot 2^{\ell}.\label{eqn:large-outside}
\end{align}
For, suppose to the contrary that the inequality (\ref{eqn:large-outside}) does not hold. Then
\begin{align*}
\sum_{z \in C_{\ell+1} \setminus C_{\ell}} d_{zv^*} = \sum_{\substack{z \in C_{\ell+1} \setminus C_\ell: \\ d_{zv^*} \leq \nicefrac{\rho}{2}}} d_{zv^*} + \sum_{\substack{z \in C_{\ell+1} \setminus C_\ell: \\ d_{zv^*} > \nicefrac{\rho}{2}}} d_{zv^*}  \leq 2^{\ell} \cdot \nicefrac{\rho}{2} +\nicefrac{\rho}{2}\cdot 2^{\ell}\cdot  1 = 2^{\ell} \cdot \rho,
\end{align*}
which contradicts (\ref{eqn:large-dist-pk}).

Next, we define 
\[O_u^{C_{\ell+1} \setminus C_{\ell} } = \{z \in C_{\ell+1} \setminus C_{\ell} \mid z \not \in \text{Ball}(v^*, \nicefrac{\rho}{2} )\},\]
so inequality (\ref{eqn:large-outside}) is equivalent to
$|O_u^{C_{\ell+1} \setminus C_{\ell} }|> \nicefrac{\rho}{2} \cdot 2^{\ell}.$ We now show that taking $S_u^C = O_u^{C_{\ell+1} \setminus C_{\ell} }$ meets the requirements for $S_u^C$ above. We know that $|Z_u^C| \leq 2^{\phi(p_{\text{first}})} \leq 2^{\ell}$, which combined with the previous inequality gives us that 
$|Z_u^C| \leq 2^{\ell} \leq \nicefrac{2}{\rho} \cdot |O_u^{C_{\ell+1} \setminus C_{\ell} }|.$

It remains to see that  $d_{uz}, 1-d_{uz} = \Omega(1)$  for $z \in O^{C_{\ell+1} \setminus C_{\ell} }_u$. For every $z \in O^{C_{\ell+1} \setminus C_{\ell} }_u$, we have that $d_{uz} \geq d_{v^*z}-d_{uv^*}  \geq \nicefrac{\rho}{2} - \varepsilon.$ 
Further, the pivot $p$ of phase $\ell$ is close to $u$ (by the case assumption) and close to $z$ (since $z$ is clustered during phase $\ell$), so $d_{uz} \leq d_{up}+d_{pz} \leq 2r+r =3r.$

\medskip

In summary, we have shown that for each case there exists a set $S_u^C \subseteq C$ with $|Z_u^C| \leq \nicefrac{4}{\rho} \cdot |S_u^C|$ such that for all $z \in S_u^C$, it holds that $d_{uz} \geq \nicefrac{\rho}{2} - \varepsilon$ and $1-d_{uz} \geq 1-\varepsilon - 2r - 3\rho \cdot (2\alpha+1)$.

\begin{align*}
\sum_{C \not \ni u} \sum_{v \in Z_u^C}(1-d_{uv})
    & \leq   \sum_{C \not \ni u} |Z_u^C|\\
    & \leq  \nicefrac{4}{\rho} \cdot \sum_{C \not \ni u} |S_u^C| \\
    & \leq  \nicefrac{4}{\rho} \cdot  \sum_{C \not \ni u}
      \Bigg (\sum_{\substack{z \in S_u^C: \\ \mathcal{C}^*(u)= \mathcal{C}^*(z)}} \frac{1}{\nicefrac{\rho}{2}-\varepsilon}\cdot d_{uz}+
    \sum_{\substack{z \in S_u^C :\\ \mathcal{C}^*(u) \neq \mathcal{C}^*(z)}} \frac{1}{1-\varepsilon - 2r - 3\rho \cdot (2\alpha+1)} \cdot (1-d_{uz})  \Bigg )\\
    & \leq  \nicefrac{4}{\rho} \cdot \max \left \{\frac{1}{\nicefrac{\rho}{2}-\varepsilon}, \frac{1}{1-\varepsilon - 2r - 3\rho \cdot (2\alpha+1)} \right\} \cdot \text{OPT}(u).
\end{align*}
Then we can sum over all $u$ to see
\begin{align*}
    \sum_{uv \in Z}(1-d_{uv}) &= \sum_{u} \sum_{C \not \ni u} \sum_{v \in Z_u^C}(1-d_{uv})\\
    &\leq   \nicefrac{4}{\rho} \cdot \max \left \{\frac{1}{\nicefrac{\rho}{2}-\varepsilon}, \frac{1}{1-\varepsilon - 2r - 3\rho \cdot (2\alpha+1)} \right\} \cdot  \sum_{u}\text{OPT}(u) \\
    &=  \nicefrac{8}{\rho} \cdot \max \left \{\frac{1}{\nicefrac{\rho}{2}-\varepsilon}, \frac{1}{1-\varepsilon - 2r - 3\rho \cdot (2\alpha+1)} \right\} \cdot  \text{OPT}.
\end{align*}

\end{proof}

\subsection{Bounding $S^-$}\label{subsec:neg}
In this subsection, we bound $S^-$ from Equation (\ref{eqn: plus-minus-part}) by (as in the case for positive edges) partitioning the sum
into 
\[S^- = \sum_{uv \in B_1}d_{uv}+\sum_{uv \in B_2}d_{uv}+\sum_{uv \in B_3}d_{uv}\]
with 
\begin{align*}
    B_1 &= \{uv \mid \mathcal{C}(u)=\mathcal{C}(v),~ \mathcal{C}^*(u) = \mathcal{C}^*(v) \}\\
    B_2 &= \{uv \mid \mathcal{C}(u)=\mathcal{C}(v),~ \mathcal{C}^*(u) \neq \mathcal{C}^*(v) \text{ and } d_{uv} \leq 1-\varepsilon \},\\
    \text{ and }
B_3&= \{uv \mid \mathcal{C}(u)=\mathcal{C}(v),~ \mathcal{C}^*(u) \neq \mathcal{C}^*(v)\text{ and } d_{uv} > 1-\varepsilon\},
\end{align*}
where again we take $\varepsilon = \nicefrac{\rho}{4}$. In words, $B_1$ are the $uv$ pairs where our clustering $\mathcal{C}$ and the optimal clustering $\mathcal{C}^*$ agree in clustering $u$ and $v$ together, whereas $B_2$ and $B_3$ contain $uv$ pairs where $\mathcal{C}$ and $\mathcal{C}^*$  disagree on whether or not to cluster $u$ and $v$ together or separate. In $B_2$, however, we can easily charge to $d$ since we cluster together pairs whose distance isn't too large, whereas in $B_3$ we need to charge to other edges. The next lemma bounds the cost incurred from edges in $B_1$ and $B_2$.

\begin{lemma}\label{lem: B1B2}
    \[\sum_{uv \in B_1}d_{uv}+\sum_{uv \in B_2}d_{uv}\leq \nicefrac{1}{\varepsilon}\cdot \text{OPT}.\]
\end{lemma}
\begin{proof}[Proof of Lemma \ref{lem: B1B2}]
  For $uv \in B_1$, both $\mathcal{C}$ and $\mathcal{C}^*$ pay the same cost on that edge.
For $uv \in B_2$, $\mathcal{C}$ pays $d_{uv} \leq 1-\varepsilon$, while $\mathcal{C}^*$ pays $1-d_{uv} \geq \varepsilon$, so $d_{uv} \leq \nicefrac{1}{\varepsilon} \cdot (1-d_{uv})$.
The cost from $uv \in B_1 \cup B_2$ is at most
\[\sum_{uv \in B_1}d_{uv}+\sum_{uv \in B_2}d_{uv} \leq \sum_{C \in \mathcal{C}^*}\sum_{u,v \in C} d_{uv}+ \sum_{C,C' \in \binom{\mathcal{C}^*}{2}}\sum_{u \in C} \sum_{v \in C'} \nicefrac{1}{\varepsilon} \cdot (1- d_{uv}) \leq \nicefrac{1}{\varepsilon}\cdot \text{OPT}.\] 
\end{proof}

It remains to bound $\sum_{uv \in B_3}d_{uv}$. 
Fix a vertex $u$, and let $C \in \mathcal{C}$ be the cluster of $u$. Let $B_3(u) = \{v \in C \mid uv \in B_3, v\text{ arrives before }u\}$. Thus, it must be the case that $u$ was clustered by a pivot, call it $p$. It is possible, on the other hand, that $v$ started $C$ and thus does not have a pivot. Take $p'$ to be $v$ in this case and to be the pivot of $v$ otherwise. First we see that it must be the case that $p$ and $p'$ are far apart.

\begin{claim}\label{claim: far-pivots}
Fix $u$ and $v \in B_3(u)$. 
For $p$ and $p'$ as defined above, we have $d_{pp'} > 1-2r-\varepsilon$.
\end{claim}
\begin{proof}[Proof of Claim \ref{claim: far-pivots}]
If $p' = v$, then we have $d_{pp'} = d_{pv} \geq d_{uv} - d_{up} > 1-\varepsilon - r$. 
Otherwise, $p'$ is the pivot of $v$, and we have 
\[d_{pp'} \geq d_{uv} - d_{up} - d_{vp'} \geq 1-\varepsilon - 2r.\]
\end{proof}

Using this claim, we bound the cost of $B_3$.
\begin{lemma}\label{lem: b3}
    \[\sum_{uv \in B_3} d_{uv} \leq \max \left \{\nicefrac{8}{13\rho}, \nicefrac{8}{1-20\rho } \right \}\cdot \text{OPT}.\]
\end{lemma}
\begin{proof}[Proof of Lemma \ref{lem: b3}]

Fix $u$ and its cluster $C \ni u$.
Let $v^*$ be the latest arriving $v \in B_3(u)$, and let $\{ p_1,\ldots,p_k\}$ be the set of pivots for $C$ between when $v^*$ is clustered and when $u$ is clustered, inclusive. 
Note this means that if $v^*$ initiates $C$, then $v^*$ has no pivot, but $v^*$ immediately becomes the (first) pivot for $C$, so $p_1 = v^*$; otherwise, $p_1$ is the pivot for $v^*$. Note also that $p_k$ is the pivot for $u$.

Let $i$ be minimal such that $d_{p_iu} \leq 2r$, which exists since one can take $p_i=p_k$, as $p_k$ clusters $u$.
Note that if $i=1$, then also $d_{p_ip_k} \geq1-2r-\varepsilon$ by Claim \ref{claim: far-pivots}, so   we have both an upper and lower bound on $d_{p_iu}$ with   $1-3r-\varepsilon \leq d_{p_iu}\leq 2r$.
Consider when $i>1$.
We claim here $d_{p_iu}\geq 2r-3 \rho$. To see why, if this were not the case, 
we would have $d_{p_{i-1}u}\leq d_{p_{i-1}p_i}+d_{p_iu} \leq 3\rho+2r-3\rho=2r$ (since $i>1$, $p_{i-1}$ exists) using Claim \ref{claim: pivot-migration}, contradicting the minimality of $i$. Thus when $i>1$, we have $2r-3\rho \leq d_{p_iu} \leq 2r$. 
Taking the more restrictive bound, we see that in both cases ($i=1$ and $i>1$), there exists $p_i$ with $2r-3\rho \leq d_{p_iu} \leq 2r$. 

Recall that $T_{p_i}$ are the nodes in $C_{\phi(p_i)}$ (the nodes in $C$ at the time $p_i$ becomes a pivot for $C$) in the ball of radius $2 \rho$ around $p_i$. Then for $z \in T_{p_i}$, $d_{zu} \leq d_{up_i} + d_{p_iz} \leq 2r+2 \rho$ and $d_{zu} \geq d_{up_i} - d_{zp_i} \geq 2r-5 \rho.$ Also by Claim \ref{claim: dense-pivot}, $|B_3(u)| \leq 2 \cdot |C_{\phi(p_i)}| \leq 4 \cdot |T_{p_i}|$.

In summary, for a fixed $u$, the pivot $p_i$ is chosen so that $T_{p_i}$ contains a large set of nodes $z$ that have distance upper and lower bounded to $u$. This allows us to bound the cost of all $uv$ for $v \in B_3(u)$ as follows.

   \begin{align*}
       \sum_{v \in B_3(u)}d_{uv}
       &\leq  |B_3(u)|\\
       & \leq 4 \cdot |T_{p_i}|\\
       &=4 \cdot \sum_{\substack{z \in T_{p_i}  :\\ \mathcal{C}^*(u)=\mathcal{C}^*(z)}}1+4 \cdot \sum_u \sum_{\substack{z \in T_{p_i} :\\ \mathcal{C}^*(u)\neq\mathcal{C}^*(z)}}1\\
       &\leq 4 \cdot \sum_{\substack{z \in T_{p_i}  :\\ \mathcal{C}^*(u)=\mathcal{C}^*(z)}}\nicefrac{1}{2r-5\rho }\cdot d_{uz}+4 \cdot \sum_{\substack{z \in T_{p_i}  :\\ \mathcal{C}^*(u)\neq \mathcal{C}^*(z)}}\nicefrac{1}{1-2r-2\rho }\cdot (1-d_{uz})\\
       & = \max \left \{\nicefrac{4}{2r-5\rho}, \nicefrac{4}{1-2r-2\rho } \right \}\cdot \text{OPT}(u).
   \end{align*}
   Then we conclude the claim since 
   \begin{align*}
        \sum_{uv \in B_3} d_{uv} &=   \sum_u \sum_{v \in B_3(u)}d_{uv} \leq \max \left \{\nicefrac{8}{2r-5\rho}, \nicefrac{8}{1-2r-2\rho } \right \}\cdot \text{OPT}.
   \end{align*}
\end{proof}

\subsection{Proof of Theorem \ref{thm:main}}

We are ready to prove our main result, that the competitive ratio of our algorithm is constant. We note that while we do not prove that our choices of constants in the statement of Algorithm \ref{alg: main-alg} are optimal, the choice of $r = \nicefrac{18}{115}$, equivalently, $\rho = \nicefrac{2}{115}$, is due to Lemma \ref{lem:z-charge}. This choice of $\rho$ balances the two terms over which we take the maximum in that lemma. The rationale for this is that the bounds in all of the other lemmas mentioned in the proof of Theorem \ref{thm:main} are dominated by the bound in Lemma \ref{lem:z-charge}. 

\begin{proof}[Proof of Theorem \ref{thm:main}]
    We recall that $r = \nicefrac{18}{115}$, $\rho = \nicefrac{r}{9}$, and $\varepsilon = \nicefrac{\rho}{4}$. Given these choices, combining Lemmas \ref{lem: A1A2}, \ref{lem:x-charge}, \ref{lem:y-charge}, \ref{lem:z-charge}, \ref{lem: B1B2}, and \ref{lem: b3} gives that the cost of Algorithm \ref{alg: main-alg} is at most $108048 \cdot \text{OPT}$.
\end{proof}

\section{Conclusion}

While the unweighted correlation clustering problem has an $\Omega(n)$ lower bound on the competitive ratio in the fully online setting against adversarial arrival order, we show that online metric correlation clustering admits a constant-competitive algorithm. Our algorithm relies on the insight that pivots should be allowed to move within a cluster, so long as consecutive pivots move gracefully (are not too far apart). A natural follow-up direction to this work is to study lower bounds; we conjecture it is impossible to design algorithms in this setting with competitive ratio better than a small constant. It also remains to improve our competitive ratio. While we did not fully optimize the choice of constants, getting the factor to a small constant will require new ideas in the charging arguments, or perhaps a different algorithm entirely.

\printbibliography

\end{document}